\documentclass[11pt]{article}
\usepackage[capposition=top]{floatrow}
\usepackage{amssymb,amsmath,amsthm,amsfonts,changepage,enumitem,eurosym,geometry,ulem,graphicx,caption,color,setspace,sectsty,comment,footmisc,caption,pdflscape,subcaption,array,xr-hyper,hyperref,pdfsync,enumerate,ragged2e,booktabs,adjustbox}
\usepackage{threeparttable}
\usepackage[square,sort&compress,sectionbib]{natbib}
\usepackage{bbm}
\usepackage{datetime}
\newdateformat{monthyeardate}{\monthname[\THEMONTH] \THEYEAR}
\allowdisplaybreaks
  \hypersetup{ 
        colorlinks   = true,
        citecolor    = blue,
        linkcolor    = blue,
        urlcolor     = black,
    }
    
\makeatletter
\newcommand*{\addFileDependency}[1]{% argument=file name and extension
  \typeout{(#1)}
  \@addtofilelist{#1}
  \IfFileExists{#1}{}{\typeout{No file #1.}}
}
\makeatother

\newtheorem{theorem}{Theorem}

\newtheorem{proposition}{Proposition}
\newtheorem{assumption}{Assumption}

\theoremstyle{plain}

\begin{document}
\setlength{\abovedisplayskip}{8pt}
\setlength{\belowdisplayskip}{8pt}

\begin{titlepage}
\title{\singlespacing Instrumental Variables with Time-Varying Exposure: \\ Dynamic Effects of Revascularization on Quality of Life\thanks{We thank Michal Kolesar and Amanda Kowalski for helpful comments; Fatima Djalalova provided exceptional research assistance. This manuscript was prepared using ISCHEMIA Research Materials obtained from the NHLBI Biologic Specimen and Data Repository Information Coordinating Center and does not necessarily reflect the opinions or views of ISCHEMIA or the NHLBI}   }
\author{Joshua D. Angrist\thanks{ MIT Department of Economics and NBER. Email: \href{mailto:angrist@mit.edu}{angrist@mit.edu}} \\   Bruno Ferman\thanks{Sao Paulo School of Economics, FGV. Email: \href{mailto:bruno.ferman@fgv.br}{bruno.ferman@fgv.br}}\and Carol Gao\thanks{ MIT Operations Research Center. Email: \href{mailto:carolgao@mit.edu}{carolgao@mit.edu}}  \\ Peter  Hull\thanks{ Brown University Department of Economics and NBER. Email: \href{mailto:peter_hull@brown.edu}{peter\_hull@brown.edu}} \and Otávio Tecchio\thanks{MIT Department of Economics. Email: \href{mailto:otecchio@mit.edu}{otecchio@mit.edu}}  
\\ Robert W. Yeh\thanks{Smith Center for Outcomes Research and Division of Cardiovascular Medicine, Department of Medicine, Beth Israel Deaconess Medical Center Email: \href{mailto:ryeh@bidmc.harvard.edu}{ryeh@bidmc.harvard.edu}}
}
\monthyeardate
\maketitle

% Abstract for NBER / Websites
\begin{abstract}
\begin{singlespace}
\noindent 
\par 

\noindent This paper develops instrumental variables (IV) estimators for dynamic causal effects in randomized trials with imperfect compliance. These methods are applied to a randomized trial that assigned patients with ischemic heart disease to either an invasive treatment arm centered on revascularization or a control group meant to receive non-invasive medical therapy. As is common in such ``strategy trials,'' many participants assigned to treatment remained untreated while many assigned to control crossed over into treatment.  Protocol non-compliance causes ITT estimates to diverge from the effect of treatment received, while conventional per-protocol analyses that condition on treatment received are compromised by selection bias. Extending the static potential-outcomes IV framework, the methods here identify average causal effects of treatment for dynamic compliers, the set of trial participants who comply with trial protocol at different follow-up horizons.  IV estimates of revascularization effects on compliers’ quality of life are markedly larger and more sustained than previously reported ITT and per-protocol estimates. We also show how to estimate average characteristics and marginal potential outcome means for dynamic compliers.  These results are used to explain confounding in as-treated per-protocol estimates.

\end{singlespace}
\end{abstract}

\onehalfspacing
\setcounter{page}{0}
\thispagestyle{empty}
\end{titlepage}
 
\section{Introduction}

Many patients with ischemic heart disease undergo revascularization, an invasive strategy involving percutaneous coronary intervention (PCI) and/or coronary artery bypass grafting (CABG). These are invasive, resource-intensive, and potentially risky procedures. Alternatively, ischemia patients may be treated conservatively with a combination of lifestyle changes and drugs. Whether the benefits of revascularization outweigh the associated risks and costs remains controversial. 

The ISCHEMIA trial aimed to assess the benefits of revascularization in a randomized clinical trial (RCT) contrasting an invasive strategy centered on revascularization with conservative treatment involving medical therapy alone. The intention-to-treat (ITT) estimates that are the focus of ISCHEMIA analyses to date compare trial participants by assigned treatment arm. These show no statistically significant difference in mortality, while ITT estimates of assignment effects on angina frequency and other Seattle Angina Questionnaire (SAQ) domain scores indicate modest and fading effects on quality of life and angina frequency \citep{maron2020,spertus2020ckd,spertus2020}.   

These findings notwithstanding, there's more to learn from the ISCHEMIA trial than has been previously reported. ITT estimates are affected by substantial treatment group non-adherence (i.e., some assigned to invasive treatment were not revascularized) along with control group crossovers (i.e., some assigned to conservative treatment were revascularized). In other words, assignment to the invasive treatment group boosted revascularization rates substantially, but not deterministically. Comparisons by treatment assigned reflect a mix of causal effects of revascularization for trial participants who comply with the trial protocol and null effects for participants whose revascularization status is unchanged by random assignment \citep{pnas2023}. Conventional per protocol estimators attempt to undo the resulting dilution by conditioning on treatment received. But treatment received is typically far from randomly assigned: patients and their doctors choose whether to revascularize, a choice determined in part by patient frailty and prognosis.  

The ISCHEMIA trial ran for five years, a period during which trial participants assigned conservative medical therapy were free to cross over to revascularization. The timing of noncompliance varied across participants in both arms of the ISCHEMIA trial. Most assigned to revascularize did so shortly after random assignment. Some control-group crossovers were also revascularized shortly after random assignment, but many waited years before invasive treatment.  Consequently, in each follow-up wave, control-group crossovers include a mix of recently-revascularized patients and patients who were revascularized several years ago. This time-varying exposure has implications for estimating dynamic treatment effects, which might emerge if revascularization benefits fade, and is common in strategy trials.\footnote{For example, the recent British Heart Foundation SENIOR-RITA trial is affected by non-adherence and time-varying crossover rates of magnitudes similar to those in ISCHEMIA \citep{kunadian2024}.}

These problems are addressed here with the method of instrumental variables (IV), using random assignment as an instrument for time-varying exposure to generate new estimates of revascularization effects on quality-of-life. Our methodological discussion starts with a theoretical result motivated by compliance behavior seen in the ISCHEMIA data. This approach exploits the fact that random assignment induces immediate revascularization for participants assigned to invasive treatment, but does not otherwise change the distribution of revascularization exposure by much. Following \cite{newkappa2025}, we show this behavior is sufficient (when added to standard IV assumptions) to identify effects of immediate exposure averaged over all possible levels of counterfactual exposure. 

We then turn to identification of incremental and cumulative average causal effects of time-varying exposure, for each possible exposure value, under a set of identifying assumptions well-suited to strategy trials like ISCHEMIA. This identification strategy exploits novel \emph{wave ignorability} assumptions that leverage the repeated measurement of outcomes across multiple follow-up waves. Both strategies identify average causal effects of revascularization exposure for patients revascularized or revascularized sooner as a result of invasive assignment. Finally, we show how to estimate mean characteristics of variables like baseline SAQ scores for trial participants---labeled \emph{dynamic compliers}---whose exposure is changed by random assignment. These theoretical results should be applicable in a wide range of settings. An appendix extends identification results for dynamic compliers to marginal distributions of potential outcomes.

Application of these new IV tools to the ISCHEMIA trial reveals quality of life gains from revascularization that are substantially greater and more sustained than previously reported. IV estimates suggest revascularization yields sustained gains in SAQ summary scores of around four points over five years, while revascularization boosts angina frequency scores by a sustained three points after an initial-year effect spike of 5-6 points. Conventional ITT estimates of revascularization effects on SAQ and angina scores, by contrast, fade to around two points over a five-year horizon. As-treated per-protocol estimates similarly fade. We show that the latter result is explained by the growing share of relatively sick control-group crossovers  and the fact that treatment-group nonadherents are an increasingly healthy group.   

Our work builds on a large literature in statistics, biostatistics, and epidemiology that develops estimators for causal effects of time-varying treatment exposure (including, for instance, \cite{robins1994snm,RobinsHernan2008TimeVaryingExposures,hernan_robins_2020}). Much of this literature relies on sequential conditional independence assumptions, such as those used in foundational g-methods and structural nested models \citep{robins1998structural,Robins1997ComplexLongitudinal,Robins2000MSMvsSNM,robins2004optimal} and related approaches \citep[e.g.,][]{Qu2020_general_framework_treatment_effects,Luo2022_adherers_multiple_imputation,OlarteParra2022_hypothetical_estimands,Rytgaard2022ContinuousTimeTMLE}. IV methods, by contrast, rely on the initial randomization of treatment offers along with exclusion and monotonicity assumptions likely to hold in RCTs. To these we add either a (testable) restriction on compliance patterns or an event-study-type restriction on average causal effects. Sequential conditional independence requires that, conditional on treatment and covariate history, subsequent treatment decisions are as good as randomized with respect to future potential outcomes. This seems unattractive in our setting, since decisions to discontinue or modify treatment over time are likely to depend on unmeasured and evolving patient characteristics (a point noted in a related context by \cite{Michael2024_IV_MSM_TimeVarying}).

The identification strategies developed here are also related to work using instrumental variables to identify causal effects in clinical trials. In trials with a static treatment and imperfect compliance, initial treatment assignment is a natural instrument for treatment received \citep{pnas2023}.\footnote{IV tools have long been applied to randomized trials in social science and medicine, though sometimes without reference to IV. For example, \cite{bloom1984} and \cite{newcombe88} derive adjustments for randomized trials with control-group crossovers while \cite{baker94} use maximum likelihood to derive IV-type adjustments for nonadherence for Bernoulli outcomes. \cite{baker2024} characterize this shared intellectual history.} Identification of dynamic causal effects, however, typically requires stronger assumptions than those needed to identify a single static causal effect \citep{hernan_per-protocol_2017,hernan2006instruments}. Our analysis differs in two ways from previous applications of IV to models with dynamic non-compliance. First, we require only a single static instrument (derived from initial random assignment) while earlier work typically relies on a time-varying instrument  \citep[see, e.g.,][]{RobinsHernan2008TimeVaryingExposures,Michael2024_IV_MSM_TimeVarying}. Second, we consider point-in-time treatment effects for outcomes like quality-of-life scores while earlier work focuses on parametric  survival models for absorbing states like death \citep[e.g.,][]{martinussen2017instrumental,shi2021instrumental,ying2023structural,Michiels2024_AdjustingTimeVaryingSwitches,Yende-Zuma2019_AdjustingEffectIntegratingARTTB,RobinsTsiatis1991RPSFTM}.\footnote{Simulation studies comparing causal effect estimators applied to survival data from randomized trials include \citet{Morden2011} and \citet{Latimer2018}. Other frameworks for time-varying effects in randomized trials include the Bayesian principal stratification models developed in \cite{lin2008} and \cite{mattei2024}.} 

In a conceptual framework closely related to ours, \citet{Bowden2025} identifies dynamic treatment effects in the context of a parametric model with a single static instrument and time-varying nonadherence. In contrast with \citet{Bowden2025}, we identify dynamic causal effects without parametric restrictions. Our analysis also differs from \cite{Bowden2025} in allowing for control-group crossovers, a signal feature of ISCHEMIA and many other RCTs with partial compliance.\footnote{Appendix \ref{appendix: bowden} extends the \cite{Bowden2025} parametric model to accommodate control-group crossovers. This appendix details limitations of the \cite{Bowden2025} parametric setup.} Also closely related is \citet{rose2021}, which explores the identifying power of a single instrument when dynamic compliers are moved on the extensive margin only (i.e., from untreated to positive treatment). Our theoretical analysis begins with an analogous assumption in which dynamic compliers are moved from positive to maximal exposure. To our knowledge, this restriction is new to appear in IV analyses of clinical trials.\footnote{\cite{ferman_tecchio_2023} explores the identifying power of assumptions like wave ignorability for a variety of dynamic treatment effect models. Assumptions limiting the dependence of causal effects on calendar time also appear in the biostatistics literature (e.g., \cite{jiangchendingbiomk2023}).}

Beyond ISCHEMIA, the identification strategies developed here are well-suited to applications in which a time-invariant instrument affects time-varying treatment exposure. Examples from economics include studies of dynamic causal effects of fertility on labor market outcomes (e.g., \cite{bronars_aertwins, lundborg2017, lundborg2024}) and time-varying effects of job training programs, as in \cite{bocarhamlalonde2017}. Recent work using in-vitro fertilization as an instrument for childbirth obtains theoretical results related to our Theorem \ref{thm:lambda} \citep{besnesIVF,gallenIVF} while \cite{vohragoldin2024} give partial identification results for the effect of going from the smallest (shortest) to the largest (or longest) treatment exposure. Wave ignorability, the link between cumulative and incremental average causal effects in a potential outcomes model identified by IV, and tools for analyzing marginal potential outcomes seem useful in these and other settings.  

The rest of this paper is organized as follows. The next section introduces the ISCHEMIA trial and describes the evolution of compliance behavior over the life of the trial. Section \ref{sec:metrics} develops new IV identification strategies, following a brief review of IV methods applied to clinical trials with a static Bernoulli treatment. This section also shows how to identify the average characteristics of dynamic compliers.  Section \ref{sec:ests} applies this theoretical framework to ISCHEMIA. Section \ref{sec:concl} summarizes and discusses broader implications.

\section{The ISCHEMIA Trial}

The ISCHEMIA trial randomized 5,179 patients with moderate-to-severe cardiac ischemia to one of two care strategies. Patients assigned to the invasive treatment arm were meant to undergo diagnostic coronary angiography and subsequent revascularization when feasible (through PCI or CABG) while also receiving medical therapy. Conservative-arm patients were meant to receive medical therapy alone, with possible invasive treatment when medical therapy was deemed inadequate \citep{maron2020, spertus2020}. The trial followed patients for up to five years after random assignment.

Our analysis of ISCHEMIA is organized around four scientific goals or questions. The first is to estimate effects of revascularization exposure on SAQ summary scores and angina frequency one to five years after randomization. The second is to determine whether these effects are short-lived or persistent. The third is to compare dynamic IV estimates of cumulative exposure effects with ITT and conventional as-treated/per-protocol analyses. Fourth, we aim to estimate the baseline characteristics of dynamic compliers---the patient group whose revascularization exposure is shifted by random assignment. 

These matters are analyzed using data collected over the five-year life of the trial. Sample sizes decrease over time since participants who enrolled in the trial later contribute fewer observations ahead of the last follow-up date in December 2018.  Roughly 4,300 participants contribute to the analysis of quality of life in wave 1, a number that falls to 670 by wave 5. The proportion with follow-up data is similar in the two arms of the trial. Falling sample sizes over time are due primarily to censoring by trial end date, as can be seen in Supplement Table 2 in \cite{spertus2020}.\footnote{Estimates that control for expected follow-up time are virtually indistinguishable.}  

Because SAQ outcomes are defined only for living trial participants observed at follow-up, the IV estimands detailed in the next section are subject to the usual concerns arising from truncation by patient death. These concerns are mitigated, however, by earlier analyses of ISCHEMIA showing that mortality is unchanged by treatment assignment (e.g., \cite{maron2020, bangalore2020, bangalore2022}). Balanced survival does not guarantee the absence of survival-related selection bias, especially for IV estimands defined by latent compliance groups. Given balanced and high follow-up rates, however, adjustments based on randomized-arm mortality differences would be expected to have little effect on  ITT (and hence IV) estimates. Our estimates can therefore be interpreted as causal effects among participants alive and observed at follow-up. 

ISCHEMIA participants assigned to the invasive arm enjoyed a moderately higher quality of life as a result. Estimated ITT effects on SAQ summary scores and angina frequency scores, plotted in Figure \ref{f:ITTeffects} and reported in Table \ref{t:sample}, show modest treatment-induced gains in each wave. These effects are estimated by differences in outcome means, with covariate adjustment.\footnote{Specifically, ITTs come from ordinary least squares regressions of outcomes on assigned treatment, controlling for baseline outcomes and enrollment regions to boost precision.} Estimated ITT effects on the summary score range from just above 3 in the first wave to 1.7 in wave 4; angina gains initially peak a little higher in the first wave and fall a little lower in wave 3. These estimates can be compared to outcome means of 86-94, with standard deviations of $13$-$15$ (see Table \ref{t:sample}).

The ISCHEMIA trial is characterized by high and variable rates of dynamic noncompliance with experimental treatment assignment, reported in Table \ref{t:sample}.  Revascularization rates in the control group increase from 12\% in wave 1 to 29\% in wave 5. At the same time, only 80\% of those assigned invasive were initially revascularized.  Revascularization rates in the group assigned invasive increased from 80\% in wave 1 to 83\% in wave 5. Importantly, the difference in the likelihood of revascularization by assignment status, reported in the third column of the table, is well below one. Moreover, this gap falls from a high of 68\% in the first wave to 54\% by wave 5. We show below that these high rates of nonadherence and crossover likely dilute the ITT estimates in Figure \ref{f:ITTeffects}. 

Changes in treatment take-up over time induce differences in revascularization exposure, defined as the number of years since revascularization. Figure \ref{f:exposure histogram} charts the distribution of exposure conditional on assignment for each wave. These plots indicate that most invasive-arm patients who were revascularized were revascularized shortly after random assignment: in each wave $w$, most of those assigned to treatment have 0 or $w$ years of exposure. The next section shows how such immediate compliance helps identify average dynamic effects.

\section{IV for Clinical Exposure}\label{sec:metrics}

\subsection{Conceptual Framework}\label{sec:concept}

Consider a setting like ISCHEMIA in which trial participants are randomized to either an invasive treatment like revascularization or a conservative alternative like medical therapy. Let $Z \in \{0,1\}$ denote a dummy variable indicating invasive assignment. Quality of life outcomes are measured in each of $\bar{w}$ annual follow-up waves; $Y_w$ denotes quality of life measured in wave $w \in \{1,2,3,\hdots, \bar{w}\}$. Treatment is assigned some time ahead of wave 1.\footnote{ISCHEMIA collected data every few months in the first year and then semi-annually \citep{spertus2020}. Our analysis is limited to relatively extensive annual follow-ups. \cite{spertus2020} uses a Bayesian framework to estimate high-frequency time-varying exposure effects on an ITT basis.} 

As we've seen, post-assignment revascularization exposure is time-varying: many participants assigned conservative were revascularized later, while some assigned invasive were never revascularized or revascularized years after random assignment. Let $T_w \in \{0,1,2,\hdots, w\}$ denote \emph{revascularization exposure}, defined as the years a patient has lived since revascularization, as seen in wave $w$. Participants in any wave are either never revascularized $(T_w=0)$, revascularized shortly after  assignment $(T_w=w)$, revascularized in the year of observation $(T_w=1)$, or revascularized at some other time in between $(1<T_w<w)$.  

Potential outcomes determine heterogeneous causal effects of revascularization in each wave. Specifically, let $Y_w(t)$ denote a participant's outcome in wave $w$ when they've lived $T_w = t$ years since revascularization. Once a participant is revascularized, they're revascularized forever.  Formally, this means that for all $w\in\{1,...,\Bar{w}-1\}$, $T_w=t$ implies $T_{w+1}=t+1$. This treatment pattern mirrors a staggered adoption setup in event-study models.\footnote{See, e.g., \cite{callaway2021, ATHEY202262, borusyak2021revisiting}.} Measured in wave $w$, the incremental individual causal effect of revascularization exposure is $Y_w(t)-Y_w(t-1)$. This notation encompasses the difference in quality of life a patient experiences when revascularized this year versus last year. 
The cumulative effect of revascularization for $t \le w$ years of exposure observed in wave $w$ is $Y_w(t)-Y_w(0)$. The cumulative effect of revascularization exposure for a participant revascularized shortly after random assignment is $Y_w(w)-Y_w(0)$.

Just as potential outcomes are indexed against treatment received, it's useful to index potential treatment against treatment assigned. To that end,
let $T_w(z)$ denote revascularization exposure in wave $w$ when $Z=z$. This is defined for all trial participants, regardless of assignment. The effect of invasive assignment on revascularization exposure in wave $w$ is therefore $T_w(1)-T_w(0)$. 

Given an exclusion restriction, random assignment makes $Z$ independent of potential outcomes and potential treatments.  This \emph{independence assumption} is formalized as: 
\begin{assumption}[Independence]\label{assump:independence}
For each wave $w\in\{1,...,\Bar{w}\}$, the random variables $Y_w(0),...,Y_w(w),T_w(1), T_w(0)$ are jointly independent of $Z$.
\end{assumption}
\noindent The exclusion restriction in this case asserts that assignment to the invasive arm affects outcomes solely by increasing the likelihood of revascularization.\footnote{Exclusion is formalized by double indexing potential outcomes as in \cite{AIR}. Let $Y_w(t,z)$ denote a participant's wave-$w$ potential outcome given $t$ years of exposure and assignment $z$. The exclusion restriction says that $Y_w(t,z)= Y_w(t,z^\prime)$  for each $t \le w$, $z$, and $z^\prime\neq z$.} This assumption is plausible in the ISCHEMIA trial, since randomization to the invasive treatment likely had no direct effects on outcomes. Importantly, the exclusion restriction allows assignment to the invasive treatment to affect outcomes via the \emph{timing} of revascularization.  

As in \cite{late94} and \cite{angrist_imbens95}, we assume that invasive assignment either induces revascularization, advances the date of revascularization, or leaves revascularization exposure unchanged. This is formalized as:
\begin{assumption}[Monotonicity]\label{assump:exposure_monotonicity}
    For each $w\in\{1,...,\Bar{w}\}$, $T_w(1) \ge T_w(0)$ almost surely.
\end{assumption}
\noindent Given monotonicity, trial participants can be classified in each wave as \emph{compliers} (with $T_w(1)>T_w(0)$), \emph{always-takers} (with $T_w(1)=T_w(0)>0$), or \emph{never-takers} (with $T_w(1)=T_w(0)=0$). For always- and never-takers, invasive assignment leaves revascularization exposure unchanged. For compliers, invasive assignment increases exposure.\footnote{Appendix \ref{appendix:vytlacil} shows that Assumption \ref{assump:exposure_monotonicity} is equivalent to assuming revascularization decisions arise from a nonparametric latent-index model.}

Invasive assignment is also assumed to generate a \emph{first stage} in wave 1, making the instrument relevant for treatment received shortly after assignment. This restriction is: 
\begin{assumption}[Relevance]\label{assump:relevance}
    $0<P[Z=1]<1$ and $E[T_1\mid Z=1]>E[T_1\mid Z=0]$.
\end{assumption}
\noindent Under Assumption \ref{assump:independence}, instrument relevance implies a positive wave-1 first-stage effect on treatment:
$$E[T_1\mid Z=1]-E[T_1\mid Z=0]=E[T_1(1)-T_1(0)] > 0.$$ 
Because wave-$1$ compliers have $T_1(1)>T_1(0)$, the wave-1 first-stage is the probability of wave-1 compliance under Assumption \ref{assump:exposure_monotonicity}. Moreover, because treatment is irreversible, Assumption \ref{assump:relevance} implies the existence of compliers such that $T_w(1)= w>T_w(0)$ in each wave. In other words, a subset of compliers in each wave, $w$, has $w$ years of revascularization exposure when assigned invasive.   

Under Assumptions \ref{assump:independence}-\ref{assump:relevance}, a simple Wald-type IV estimand using wave-1 data identifies the average causal effect of one year of revascularization exposure for wave-1 compliers. Specifically, in wave 1, revascularization exposure, $T_1$, is a Bernoulli treatment that indicates participants revascularized shortly after random assignment. The \cite{late94} local average treatment effect (LATE) theorem applied to wave-1 data implies that:
\begin{equation}
\frac{E[Y_1\mid Z=1]-E[Y_1\mid Z=0]}{E[T_1\mid Z=1]-E[T_1\mid Z=0]}=E[Y_1(1)-Y_1(0)\mid T_1(1)>T_1(0)].\label{eq:late}
\end{equation}
The causal parameter on the right-hand side of this expression is the average effect of revascularization exposure, $Y_1(1)-Y_1(0)$, on wave-1 compliers. This is obtained by dividing the wave-1 \emph{reduced form} by the wave-1 first stage.  The former is defined as the difference in average quality of life between participants assigned invasive and participants assigned conservative, measured in the first wave. The first stage in the Wald denominator is given by the corresponding difference in wave-1 revascularization rates. 

The reduced form is an ITT effect that compares outcomes by treatment assigned, that is, by the instrument, $Z$. The first stage is the difference in treatment rates by assignment status; this is also the share of the trial population for whom treatment status is changed by random assignment.  Since this share is between zero and one, ITT effects in a static or short-run analysis of trial data are diluted relative to the effect of revascularization itself. Intuitively, because revascularization assignment is assumed to affect outcomes solely by inducing revascularization among compliers, ITT is diluted by averaging in causal effects of zero for always- and never-takers. 

Conventional as-treated analyses of clinical trials ignore random assignment. Rather, as-treated analyses compare outcomes conditional on treatment received, typically with covariate controls (see, for instance, the \cite{cabanaPP} analysis of the CABANA trial comparing treatments for atrial fibrillation). Our as-treated estimates come from ordinary least squares (OLS) regressions of $Y_w$ on $T_w$, with a few baseline controls. Non-random treatment take-up likely biases such estimates; in the LATE framework, this bias can be understood as stemming from differences in health between compliers, always-takers, and never-takers \citep{angrist2004,aronow2013}. We substantiate this view with a comparison of group characteristics in Section \ref{sec:popchar} below.

\subsection{IV for Immediate Exposure Effects}

In randomized with repeated follow-ups, random assignment changes the distribution of time-varying exposure. Monotonicity requires only that exposure be weakly increasing with invasive assignment. As Figure \ref{f:exposure histogram} suggests, however, most ISCHEMIA participants assigned invasive were revascularized in the first year following assignment, with few induced by random assignment to revascularize in later years. Participants assigned conservative, by contrast, flow into revascularization over the life of the trial. This empirical observation motivates an identifying assumption requiring compliers who are assigned invasive to be revascularized shortly after random assignment. This \emph{immediate compliers only} (IMCO) assumption is formalized as:
\begin{assumption}[IMCO]\label{assump:imco}
    Assignment either shifts participants into treatment immediately or has no effect: in each wave $w\in\{1,...,\Bar{w}\}$, $P[T_w(1) = w \mid T_w(1) > T_w(0)] = 1$. 
\end{assumption}

Assumption \ref{assump:imco} says that everyone for whom revascularization exposure is changed by random assignment is revascularized in the first wave. The following result shows that this condition is sufficient to identify LATE averaged over immediate complier types. In particular, instrumenting $T_1$ with $Z$ in wave $w$ identifies a causal effect averaged over all immediate compliers seen in wave $w$. 
\begin{theorem}[Immediate-Exposure IV Under IMCO]\label{th:causaleffectIMCO}
      Suppose Assumptions \ref{assump:independence}-\ref{assump:relevance} and IMCO hold. Consider the IV estimator that, for a given $w$, uses $Z$ to instrument $T_1$ in 
\begin{equation}
\label{eq:initial-esxposure_linearmodel}
Y_w= \gamma_w + \tau_w T_1 + u_w,
\end{equation}
where $\gamma_w$ and $\tau_w$ are parameters and the IV estimand is defined by $E[u_w|Z]=0$. Then, coefficient $\tau_w$ in this IV estimand equals
\begin{equation}\label{eq:tau_IMCO} 
    \tau_w = \sum_{t=0}^{w -1}\omega_{wt} E[Y_w(w) - Y_w(t) \mid T_w(1)=w, T_w(0)=t],
    \end{equation}
where $\omega_{wt}\equiv P[T_w(1)=w, T_w(0)=t \mid T_1(1)>T_1(0)]$ and $\sum_{t=0}^{w-1}\omega_{wt}=1$. Moreover, for each $t< w$, $\omega_{wt}$ is identified by:
\begin{equation}\label{eq:omega}
    \omega_{wt} = -\frac{E[\mathbf 1[T_w=t] \mid Z=1] - E[\mathbf 1[T_w=t] \mid Z=0]}{E[T_1 \mid Z=1] -E[T_1\mid Z=0]}.
\end{equation}
This is an IV estimand using $Z$ to instrument $T_1$ with dependent variable $-\mathbf 1[T_w=t]$.
\end{theorem}
\begin{proof}
    See Appendix \ref{appendix:causalfxIMCO_proof}. 
\end{proof}

Estimand $\tau_w$ in Theorem \ref{th:causaleffectIMCO} averages effects of increasing exposure from $t$ to $w$ for each $t<w$. For instance, instrumenting $T_1$ in wave 3 averages cumulative effects from $0\rightarrow3,1\rightarrow3$, and $2\rightarrow3$. We refer to $\tau_w$ as an \emph{immediate exposure effect} because $t=w$ in any wave $w$ implies exposure in wave $1$.\footnote{Theorem \ref{th:causaleffectIMCO} is implied by \cite{newkappa2025} Theorem 2.3. Appendix A.1 gives a direct proof.} The immediate exposure estimand in \eqref{eq:tau_IMCO} provides a useful summary measure of immediate revascularization. This estimand does not, however, reveal how revascularization effects evolve with increasing exposure. We turn next to estimands that identify these dynamic effects.\footnote{IMCO is similar in spirit to the \cite{rose2021} restriction allowing extensive margin compliers only (EMCO). With an ordered treatment under EMCO, everyone whose exposure is increased by assignment moves from $T_w(0)=0$ to positive $T_w(1)$.} 

\subsection{Identification and Estimation of Dynamic Causal Effects}\label{sec:dynamicfx}

Dynamic exposure effects are identified with the aid of an assumption that limits systematic heterogeneity in treatment effects across waves:
\begin{assumption}[Wave Ignorability]\label{assump:wave_ignore}
    Average incremental effects for compliers are common across waves: for each pair of waves $w$ and $v\le w$, and for each exposure time $t\le v$,
 \begin{equation}\label{eq:definelambda}
        E[Y_w(t)- Y_w(t-1) \mid T_w(1) \ge t > T_w(0)] = E[Y_{v}(t)- Y_{v}(t-1) \mid T_{v}(1) \ge t > T_{v}(0)] \equiv \lambda_t.
    \end{equation}
\end{assumption}

This assumption restricts and defines incremental exposure effects for compliers, denoted $\lambda_t$. In particular, wave ignorability allows $\lambda_t$ to vary freely with exposure time while requiring incremental effects to be independent of the wave in which they're seen. For instance, in wave $w$, the initial incremental effect of revascularization, $Y_w(1)-Y_w(0)$, may differ from the incremental effect of one year of early exposure, $Y_w(w)-Y_w(w-1)$. Assumption \ref{assump:wave_ignore} also allows effects to differ for compliers, always-takers, and never-takers in each wave. Wave ignorability mirrors restrictions implicit in conventional event study regression models, which typically index dynamic causal effects by event time rather than calendar time or treatment cohort (see, e.g., \cite{miller2023} for a survey). 

The following theorem shows that, in the context of the LATE framework, wave ignorability is sufficient to identify the set of incremental causal effects, $\lambda_t$.   
\begin{theorem}[Dynamic Incremental Causal Effects]\label{thm:lambda}
    For each $t\in\{1,...,\bar{w}\}$, let $D_{wt}\equiv\mathbf{1}[T_w\ge t]$ indicate exposure of at least $t$ periods as of wave $w$. Stack the data across waves and consider the two-stage least squares (2SLS) estimator that uses $Z$ and $Z\times\mathbf{1}[j=w]$ with $j\in\{2,...,\bar{w}\}$ to instrument the set of $D_{wt}$ in the linear model
    \begin{equation}\label{eq:lambda_IV}
        Y_w = \mu + \sum_{j=2}^{\bar{w}}\alpha_j\mathbf{1}[j=w] + \sum_{t=1}^{\bar{w}} \lambda_t D_{wt} + \eta_w,
    \end{equation}
   where $\mu$ and $\alpha_j$ are constants and the corresponding 2SLS estimand is defined by $E[\eta_w\mid Z]=0$ for all $w$. Given Assumptions \ref{assump:independence}-\ref{assump:relevance} and \ref{assump:wave_ignore}, the vector of $\lambda_t$ defined in expression \eqref{eq:definelambda} equals the set of $\lambda_t$ in equation \eqref{eq:lambda_IV}.
\end{theorem}     
\begin{proof}
See Appendix \ref{appendix:lambda_proof}.
\end{proof}
      
\noindent Theorem \ref{thm:lambda} is proved using the \citet{angrist_imbens95} average causal response (ACR) theorem. The ACR theorem shows that the wave-specific IV estimand using $Z$ to instrument $T_w$ with outcome $Y_{w}$ can be written as a weighted average of incremental complier revascularization effects for each period up to $w$. Under wave ignorability, these incremental effects are common across $w$ and can therefore be obtained by solving a system of linear equations linking wave-specific reduced forms and first stages. The 2SLS estimator described in Theorem \ref{thm:lambda} implements this solution.\footnote{The ACR theorem identifies a weighted average of incremental causal effects for partially overlapping groups of compliers, while Theorem \ref{thm:lambda} identifies average incremental effects for a fixed subpopulation of compliers immediately induced into treatment. To see this, note that for any $t\le w$, wave ignorability and irreversible treatment imply
$\lambda_t=E[Y_w(t)- Y_w(t-1) \mid T_w(1) \ge t > T_w(0)] = E[Y_t(t)- Y_t(t-1) \mid T_t(1) = t > T_t(0)]=E[Y_t(t)- Y_t(t-1) \mid T_1(1)>T_1(0)]$. The second equality here is a consequence of wave ignorability and the fact that $T_t(1) \geq t$ if and only if $T_t(1)=t$; the third equality comes from the fact that, given irreversibility, $T_t(1)=t$ if and only if $T_1(1)=1$.}

Beyond average incremental effects of an additional year of exposure, it's useful to know the average causal effect of each level of exposure relative to a common reference outcome with no exposure---for ISCHEMIA, this means no revascularization. These cumulative exposure effects can be obtained under a somewhat stronger version of Assumption \ref{assump:wave_ignore}:
\begin{assumption}[Strong Wave Ignorability]\label{assump:strong_wave_ignore}
    Average incremental effects for compliers at a given level of exposure are common across waves for compliers at any possible margin: for each $t,v,w,v^\prime,$ and $t^\prime$ such that $t\le v\le w$ and $t\le t^\prime\le v^\prime\le w$, 
    \begin{equation*}
        E[Y_w(t)- Y_w(t-1) \mid T_w(1) \ge t' > T_w(0)] = E[Y_{v}(t)- Y_{v}(t-1) \mid T_{v'}(1) \ge t' > T_{v'}(0)] \equiv \lambda_t.
    \end{equation*}
\end{assumption}
 
Assumption \ref{assump:strong_wave_ignore} implies  Assumption \ref{assump:wave_ignore} (to see this, set $v'=v$ and $t'=t$ in Assumption \ref{assump:strong_wave_ignore}). In general, however, Assumption \ref{assump:strong_wave_ignore} adds to Assumption \ref{assump:wave_ignore} the requirement that a given incremental causal effect be the same at each follow-up wave for complier groups defined over every relevant level of exposure (indexed by $t'$), not just the level at which the incremental effect in question is measured (indexed by $t$). Like Assumption \ref{assump:wave_ignore}, however, Assumption \ref{assump:strong_wave_ignore} asserts that the period or wave in which effects are measured is irrelevant for incremental causal effects at a given level of exposure.      

Under strong wave ignorability, average cumulative causal effects are identified and equal to a 2SLS estimand described in the following theorem.  
\begin{theorem}[Dynamic Cumulative Causal Effects]\label{corollary:Lambda}
Under Assumptions \ref{assump:independence}-\ref{assump:relevance} and \ref{assump:strong_wave_ignore}:
\begin{itemize} 
\item[i.] For all $w$ and $t\le w$, dynamic cumulative causal effects are given by
 \begin{equation}
    \label{eq:corrol_def_Lambda}
       E[Y_w(t)- Y_w(0) \mid T_w(1) \ge t > T_w(0)]=\sum_{i=1}^t \lambda_i\equiv \Lambda_t.
    \end{equation} 
\item[ii.] For each $t\in\{1,...,\bar{w}\}$, let $R_{wt}\equiv\mathbf{1}[T_w=t]$. Stack the data across waves and consider a 2SLS estimator that uses $Z$ and $Z\times\mathbf{1}[j=w]$ with $ j\in\{2,...,\bar{w}\}$ to instrument the set of $R_{wt}$ in the linear model
    \begin{align}\label{eq:Lambda_IV}
    Y_w = \phi + \sum_{j=2}^{\bar{w}}\delta_j\mathbf{1}[j=w] + \sum_{t=1}^{\bar{w}} \Lambda_t R_{wt} + \varepsilon_w,
    \end{align}
    where $\phi$ and $\delta_j$ are constants and the 2SLS estimand is defined by $E[\varepsilon_w\mid Z]=0$ for all $w$. Then the vector of $\Lambda_t$ defined in \eqref{eq:corrol_def_Lambda} equals the set of $\Lambda_t$ in equation \eqref{eq:Lambda_IV}.
\end{itemize}
\end{theorem}
\begin{proof}
See Appendix \ref{appendix:Lambda_proof}.
\end{proof}

The first part of Theorem \ref{corollary:Lambda} says that cumulative revascularization effects are identified by the sum of incremental revascularization effects. Strong wave ignorability is key to this result: each $\lambda_t$ captures incremental causal effects for different complier groups. A scenario where trial participants have 0-2 years of exposure highlights the role played by Assumption \ref{assump:strong_wave_ignore} in this context. In this case, the relevant incremental effects in wave 2 are $\lambda_1 = E[Y_{2}(1) - Y_{2}(0) \mid T_{2}(1) \ge 1 > T_{2}(0)]$ and $\lambda_{2} = E[Y_{2}(2) - Y_{2}(1) \mid T_{2}(1) \ge 2 > T_{2}(0)]$, which average over different complier groups. The appendix proof shows that under strong wave ignorability, effects for these groups are the same at a given increment: $E[Y_{2}(1) - Y_{2}(0) \mid T_{2}(1) \ge 1 > T_{2}(0)]=E[Y_{2}(1) - Y_{2}(0) \mid T_{2}(1) \ge 2 > T_{2}(0)]$. Hence, $\lambda_1+\lambda_2=\Lambda_2$.\footnote{\label{fn:bloom}Equality of incremental causal effects across complier groups simplifies $\Lambda_t$ in settings with no control-group crossovers in the first wave. By setting $v=v^\prime=t^\prime=w$ in Assumption \ref{assump:strong_wave_ignore} we have $\lambda_t=E[Y_w(t)-Y_w(t-1)\mid T_w(1)\ge w > T_w(0)]=E[Y_w(t)-Y_w(t-1)\mid T_1(1)> T_1(0)]$ since treatment is irreversible; hence $\Lambda_t=\sum_{i=1}^t\lambda_t=E[Y_w(t)-Y_w(0)\mid T_1(1)>T_1(0)]$. This shows that, when $P[T_1(0)=1]=0$, cumulative effect $\Lambda_t$ equals the average effect of $t$ periods of exposure on the immediately-treated.} The second part of Theorem \ref{corollary:Lambda} shows that this summation is implemented by instrumenting indicators for the full set of exposure levels, $R_{wt}$.

\subsection{Characterizing Dynamic Compliers and Always-Takers}\label{sec:character}

Are IV estimates of causal effects clinically relevant? This question is answered in part by describing the baseline health (and other characteristics) of different sorts of compliers. In particular, complier characteristics inform assessments of the external validity of IV estimates. IV estimates are highly relevant when the associated complier characteristics are similar to those of a population of interest, typically the study population in an RCT. Moreover, because the treated population includes always-takers, the contrast between complier and always-taker characteristics illuminates the biases we should expect in as-treated per-protocol estimates. 

In a dynamic setting, covariate means for compliers and always-takers are identified by the following theorem, which is implied by Theorem \ref{thm:late_takers_general} in Appendix \ref{appendix:always_taker_mean_general}. The appendix theorem covers mean potential outcomes as well as means for never-takers.

\begin{theorem}[Dynamic Characterizations]\label{thm:late_takers}
    Let $X$ denote a baseline covariate, fixed across waves and unchanged by assignment. Suppose that $(X, T_w(1),T_w(0))$ is independent of $Z$ for all $w\in\{1,...,\Bar{w}\}$ and that Assumptions \ref{assump:exposure_monotonicity}-\ref{assump:relevance} hold. Then:
\begin{itemize} 
    \item[i.] \emph{Immediate complier means} are given by
    \begin{align}
    \label{eq:immediate_complier_X}
    E[X \mid T_1(1)>T_1(0)] 
    =\frac{E[\mathbf{1}[T_1=1]\times X\mid Z=1] - E[\mathbf{1}[T_1=1]\times X\mid Z=0]}{E[\mathbf{1}[T_1=1]\mid Z=1] - E[\mathbf{1}[T_1=1]\mid Z=0]}.
    \end{align}
    \item[ii.] \emph{Immediate always-taker means} are given by
    \begin{equation}
    \label{eq:immediate_AT_X}
    E[X \mid T_1(1)=T_1(0)=1]=E[ X \mid T_1=1, Z=0].
    \end{equation}
\end{itemize}
     \noindent Moreover, if Assumption \ref{assump:imco} also holds, we have that:
     \begin{itemize}
     \setcounter{enumi}{2}
     \item[iii.] Complier means are constant across waves and equal to immediate complier means:
    \begin{equation}
    \label{eq:latercomplier_ag_X}
    E[X \mid T_w(1)>T_w(0)]=E[X\mid T_1(1)>T_1(0)]
    \end{equation}
    for all $w\in\{1,...,\Bar{w}\}$;
    \item[iv.] For each $w>1$ and $t\in\{0,\dots,w-1\}$ such that
    $$E[\mathbf{1}[T_w=t]\mid Z=1]-E[\mathbf{1}[T_w=t]\mid Z=0]\neq 0, $$ \emph{disaggregated complier means} are given by
    \begin{equation}
    \label{eq:latercomplier_dsg_X}
    E[X \mid T_w(1)=w, T_w(0)=t]=
    \frac{E[\mathbf{1}[T_w=t]\times X\mid Z=1]-E[\mathbf{1}[T_w=t]\times X\mid Z=0]}{E[\mathbf{1}[T_w=t]\mid Z=1]-E[\mathbf{1}[T_w=t]\mid Z=0]};
    \end{equation}
    \item[v.] For each $w>1$, \emph{later always-taker means} are given by
    \begin{equation}
    \label{eq:laterATs_X}
    E[X\mid w>T_w(1)=T_w(0)\ge 1]
    = E[X\mid w>T_w\ge1, Z=1].
    \end{equation}
   Finally, \emph{marginal always-taker means}, which average immediate and later always-takers, can be obtained using
    \begin{equation}
    \label{eq:AT_X}
    \begin{split}
   & E[X\mid T_w(1)=T_w(0)\ge 1] \\
    &= \pi_w E[ X \mid T_1=1, Z=0]
    +
    (1-\pi_w)E[X\mid w>T_w\ge1, Z=1],
    \end{split}
    \end{equation}
    where $\pi_w=E[\mathbf{1}[T_w=w]\mid Z=0]/ \{E[\mathbf{1}[T_w=w]\mid Z=0] + E[\mathbf{1}[1\le T_w<w]\mid Z=1]\}$ is the share of immediate always-takers among all always-takers.
\end{itemize}
\end{theorem}
\begin{proof}
    See Appendix  \ref{appendix:always_taker_mean_general}.
\end{proof}

Since identification of average baseline characteristics for immediate compliers and immediate always-takers is analogous to identification of characteristics in a static IV setup, parts (i) and (ii) of Theorem \ref{thm:late_takers} follow from \cite{imbensrubin97} and \cite{abadie2003}.  The rest of the theorem uses IMCO (Assumption \ref{assump:imco}) to identify complier and always-takers means in a dynamic framework. Specifically, equation \eqref{eq:latercomplier_dsg_X} disaggregates compliers based on the exposure level they  attain when assigned conservative.  These complier means can be computed using an IV estimand that takes $\mathbf{1}[T_w=t]\times X$ as the outcome while instrumenting $\mathbf{1}[T_w=t]$. Equation \eqref{eq:laterATs_X} recovers average baseline characteristics of those revascularized after the first wave regardless of treatment assignment. Equation \eqref{eq:AT_X} combines \eqref{eq:immediate_AT_X} and \eqref{eq:laterATs_X} to recover the average characteristics of all always-takers by wave $w$. 

\section{Revascularization Effects on Quality of Life}\label{sec:ests}

We estimate causal effects of revascularization on quality-of-life outcomes by 2SLS, clustering standard errors on patient ID when waves are stacked.\footnote{Appendix \ref{sec:efficiency} establishes semiparametric efficiency of the sample analog of the 2SLS (IV) estimand in Theorem \ref{thm:lambda} (proofs for Theorems \ref{th:causaleffectIMCO} and \ref{corollary:Lambda} are similar). Our empirical specifications add baseline covariates (enrollment region and baseline angina frequency scores) to boost precision. Appendix \ref{appendix:simulations} presents a simulation experiment comparing statistical properties of IV, OLS, and a g-method estimator for a data generating process modeled on ISCHEMIA.} The assumptions underpinning this analysis are that treatment assignment is independent of potential outcomes and potential exposure and that some participants are induced to revascularize immediately by assignment to the invasive arm of the trial. The exclusion restriction implicit in Assumption \ref{assump:independence} rules out scenarios in which assignment to the invasive strategy improves quality of life with no change in revascularization. This seems plausible in an RCT focused on revascularization as the prelude to additional cardiac care. Patients assigned to the invasive arm were also much more likely to be revascularized in the first year, ensuring the assignment instrument is relevant.

Monotonicity (Assumption \ref{assump:exposure_monotonicity}) restricts the effect of random assignment on revascularization exposure. Importantly, however,  monotonicity allows unobserved determinants of revascularization including symptoms, comorbidities, quality-of-life history, physician idiosyncrasies and so on.  Monotonicity precludes only scenarios in which assignment to the invasive strategy delays or prevents revascularization---a scenario that seems unlikely.  This restriction is supported by the ISCHEMIA trial protocol and by the exposure distributions in Appendix Figure \ref{f:stoch_dom}, which shows that being assigned invasive induces earlier revascularization for all exposure values by wave 5. This first order stochastic dominance of the distribution of exposure for those assigned invasive relative to those assigned control is a testable implication of Assumption \ref{assump:exposure_monotonicity}.  

\subsection{Immediate Exposure Effects under IMCO}

Theorem \ref{th:causaleffectIMCO} offers a natural starting point for IV estimation of the causal effects of revascularization exposure. Assuming all compliers are induced to revascularize immediately (Assumption \ref{assump:imco}), IV using random assignment to instrument $T_1$ identifies $\tau_w$, an average causal effect for all complier groups. Figure \ref{f:immediate exposure} plots IV estimates of $\tau_w$ (implemented by 2SLS), along with the corresponding OLS estimates of immediate-exposure effects; these are computed by regressing SAQ outcomes on an immediate-exposure dummy in each wave in regression models that include the covariates used for IV. In this context, OLS is an as-treated per-protocol estimator that discards random assignment and is therefore susceptible to unobserved confounding.

IV estimates of immediate-revascularization effects are mostly larger than the ITT estimates in Figure \ref{f:ITTeffects} and Table \ref{t:sample}. ITT estimates for summary scores range from $1.7-3.2$ points; for angina frequency, ITT estimates fall from a high of almost 3.7 in wave $1$ to $1.8$ or less in later waves. 2SLS estimates of SAQ summary score effects, by contrast, start around $4$ and return to this level in wave $5$. Angina frequency gains decline from around 5.5 in the first wave to a low of 2 in wave 3, increasing to around 3 in waves $4$ and $5$.

OLS (as-treated) estimates of immediate-revascularization effects fall over time, declining to a level close to and statistically indistinguishable from zero by wave $5$. Divergence between OLS and 2SLS estimates signals unobserved confounding in the as-treated analysis. Differences in underlying health between compliers and others are documented in Figure \ref{f:potential means}, which plots average potential outcomes for \emph{immediate} compliers, always-takers, and never-takers (defined by $T_1(z)$) by wave. These estimates use Appendix Theorem \ref{thm:late_takers_general}.\footnote{Theorem \ref{thm:late_takers_general} encompasses Theorem \ref{thm:late_takers}, below. The latter is used in Section \ref{sec:popchar} to assess the clinical relevance of IV estimates for dynamic compliers.}

Complier mean SAQ summary and angina frequency scores are precisely estimated at around 90 and 95, respectively. Always- and never-takers are mostly sicker than compliers, with always-takers generally the worst off (estimated non-complier means are less precise than estimated complier means because non-complier groups comprise much smaller shares of the study population). Never-taker average health approaches that of compliers by wave $5$. This likely reflects the fact that sicker patients are likely to be revascularized even when assigned conservative, while patients assigned invasive who go without revascularization are an increasingly healthy group. 

\subsection{Dynamic Effects under Wave Ignorability}

2SLS estimates of cumulative exposure effects, plotted in Figure \ref{f:stacked}, suggest that revascularization generates a persistent improvement in quality of life. These estimates were computed using equation \eqref{eq:Lambda_IV} in Theorem \ref{corollary:Lambda}. Estimated exposure effects on SAQ summary scores are remarkably stable at around four. Revascularization is estimated to improve angina scores by about 5.5 points initially, with long-run gains sustained at around 3 points. Dotted lines in the figure mark average estimated exposure effects on summary scores over waves 1-5 and average estimated exposure effects over waves 2-5 for angina scores. For summary scores, a chi-square test of the hypothesis of equal cumulative effects yields a p-value of $0.14$. For angina frequency, a chi-square test of constant effects in waves 2-5 generates a p-value of $0.86$.\footnote{A test of the null hypothesis of equal angina score effects for all five exposure values generates $\chi^2(4)= 36$, a decisive rejection.}

2SLS estimates of exposure effects are markedly larger---and statistically distinguishable from---the corresponding OLS estimates. 2SLS estimates, OLS estimates, and the difference between them along with associated standard errors appear in Table \ref{t:hausman} for both SAQ outcomes.\footnote{Unlike the specification test comparing IV and OLS estimates introduced in \cite{hausman1978}, these standard errors are computed allowing unrestricted correlation between estimators.}  As can be seen in columns 3 and 6 of the table, differences by type of estimate are statistically significant at conventional levels for most of the summary score estimates and for one of the angina score estimates. For both outcomes, joint tests of equality generate decisive rejections.  It's also noteworthy that all of the 2SLS estimates in the table are significantly different from zero while the wave-$5$ OLS estimates are not.

IV estimates computed using Theorem \ref{corollary:Lambda} are robust to empirically relevant deviations from strong wave ignorability. This is demonstrated in Appendix \ref{appendix:sensitivity}, which shows that an extended model allowing for differences in average cumulative exposure effects across waves generates estimates similar to those reported in Table \ref{t:hausman}.

\subsection{Contrasting Treated Populations}\label{sec:popchar}

As a brief illustration of Theorem \ref{thm:late_takers}, Table \ref{t:compliers} compares mean baseline SAQ scores for the full sample, compliers, and always-takers. Complier means are computed using IMCO; by Theorem \ref{thm:late_takers}(iii)), these are immediate complier means, with differences over time due to attrition. The immediate and marginal always-taker means reported in columns 4-5 of Table \ref{t:compliers} are computed using Theorem \ref{thm:late_takers}(ii) and Theorem \ref{thm:late_takers}(v).  Marginal always-taker means average the baseline health of immediate and later always-takers. 

As is apparent from the first two columns in Table \ref{t:compliers}, mean baseline SAQ scores for compliers are virtually indistinguishable from average baseline scores in the full study sample. This suggests IV estimates of revascularization effects are likely to be relevant for the ISCHEMIA study population. And, consistent with the average potential outcomes plotted in Figure \ref{f:potential means}, estimated mean baseline SAQ scores for always-takers are substantially lower than complier means. 

Estimated immediate always-taker means in column 4 are similar to the marginal always-take means in column 5. Disaggregated complier means are omitted since over 80\% of the compliers allowed under IMCO move from no exposure to immediate exposure in all waves; similarly, we do not report means for the relatively small group of later always-takers.\footnote{The share of compliers with $T_w(0)=t$ as a share of all compliers is identified for all $t<w$ by equation \eqref{eq:tau_IMCO} in Theorem \ref{th:causaleffectIMCO}. Similar to equation \eqref{eq:laterATs_X}, the share of later always-takers (i.e., those with $1\le T_w(1)=T_w(0)<w$) is identified by $E[\mathbf 1[1\le T_w<w]\mid Z=1]$ under IMCO. To see this, note that $E[\mathbf 1[1\le T_w<w]\mid Z=1]=E[\mathbf 1[1\le T_w(1)<w]]=P[1\le T_w(1) = T_w(0) <w]$, where the first equality follows from independence and the second from $T_{w}(1)<w\implies T_w(1)=T_w(0)$ under IMCO.} The gap between complier and always-taker baseline health is around 10 points, comparable to wave-1 standard deviations of baseline scores of around 19 points.

\section{Summary and Conclusions}\label{sec:concl}

This paper develops instrumental variables methods for settings with dynamic treatment exposure and outcomes measured in repeated follow-up waves. Variation in exposure time and repeated follow-ups feature in countless randomized trials.  In strategy trials and other pragmatic trials, exposure variation is outside trialists' control once the assignment die has been cast. Initial assignment notwithstanding, treatment exposure is not random.

Dynamic treatment effects have long posed a challenge for clinical research. Treatment choices made after random assignment are likely correlated with post-trial potential outcomes. Our instrumental variables framework identifies average causal effects as a function of exposure time without imposing either conditional independence of treatment exposure and potential outcomes or constant treatment effects. The identification results developed and applied here exploit novel identifying assumptions that we've called IMCO and wave ignorability. IMCO is motivated by compliance dynamics observed in ISCHEMIA data; wave ignorability mirrors event-study models that presume exposure effects are unrelated to calendar time. Along with standard IV assumptions, IMCO and wave ignorability deliver causal effects of time-varying exposure that are easily estimated by 2SLS.

When applied to ISCHEMIA trial data, these tools reveal quality of life improvements that are substantially greater and more persistent than previously reported ITT and per-protocol estimates would suggest.  IV estimates imply revascularization of patients with ischemic heart disease yields sustained gains in SAQ summary scores on the order of four points. Moreover, after an initial bump up to 5.5, IV estimates of effects on angina frequency scores are stable at around three points.

In view of these estimates, it's noteworthy that conventional as-treated per-protocol (OLS) estimates decline markedly over time. Dynamic latent-group characterizations reveal this to be an artifact of the poor health of always-takers and improving health of never-takers. In particular, control group crossovers who are revascularized regardless of assignment are far sicker than compliers. At the same time, never-takers, who forgo revascularization when assigned invasive, constitute an increasingly healthy group. Consequently, the gap in SAQ scores by treatment received shrinks.

The finding that as-treated per-protocol estimates are misleading is not unique to ISCHEMIA. In an analysis of data from a randomized mammography screening trial, \cite{kowalski2023behaviour} shows mammography always-takers to be healthier than screening compliers, leading as-treated (i.e., as-screened) estimates to exaggerate gains from screening.\footnote{\cite{einavetal2020} likewise documents nonrandom selection in mammography.}

Finally, our ISCHEMIA results weigh against the view that IV complier populations are likely to be unrepresentative of any clinical population of interest (see, for instance, \cite{hernan_per-protocol_2017}). The baseline health of ISCHEMIA compliers is similar to that of the overall study population. Of course, this is an empirical result and not a theorem. But this finding too is not unique to ISCHEMIA.  In a static IV analysis of colorectal cancer screening trials, \cite{pnas2023} report mean complier characteristics much like those of the population invited for cancer screening. We expect the IV tools developed here and elsewhere to generate new, clinically useful findings in a range of applications to come.   

 %%%%%%%%%%%%%%%%%%%%%%%%%%%%%%%%%%%%%%%%%%%%%%%%%%%%%%
%%% Bibliography %%%%%%%%%%%%%%%%%%%%%%%%%%%%%%%%%%%%%
%%%%%%%%%%%%%%%%%%%%%%%%%%%%%%%%%%%%%%%%%%%%%%%%%%%%%%

%\newpage
\singlespacing
\setlength\bibsep{4pt}
\bibliographystyle{bib/aea} {\footnotesize
\bibliography{bib/references.bib} }

\clearpage 

%%%%%%%%%%%%%%%%%%%%%%%%%%%%%%%%%%%%%%%%%%%%%%%%%%%%%%
%%% Exhibits %%%%%%%%%%%%%%%%%%%%%%%%%%%%%%%%%%%%%%%%%
%%%%%%%%%%%%%%%%%%%%%%%%%%%%%%%%%%%%%%%%%%%%%%%%%%%%%%

\pagebreak
\onehalfspacing
\section*{} \label{sec:exhibits}
\addcontentsline{toc}{section}{Figures and Tables}
\renewcommand{\thesubfigure}{\roman{subfigure}}
\captionsetup[subfigure]{font=scriptsize,labelfont=scriptsize}

%%%%% F: ITT %%%%%
\begin{figure}
    \caption{Intention-to-treat Effects by Wave}
    \label{f:ITTeffects}
    \centering
    \begin{minipage}{16.5cm}        {\centering{\includegraphics[width=\linewidth]{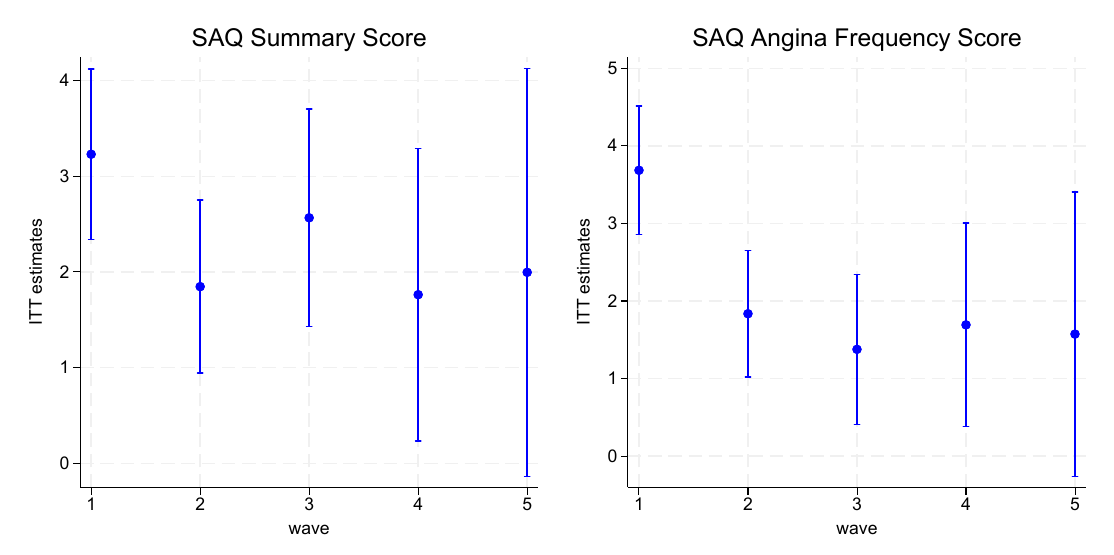}}\par}
        \footnotesize \textit{Notes:} This figure reports intention-to-treat (ITT) estimates from the ISCHEMIA trial. These contrast average SAQ scores by treatment assigned. Both outcome variables have means in the range 86-94, with a standard deviation of 13-15. SAQ scores are measured at the time of follow up. These estimates control for baseline angina frequency scores (as in \cite{spertus2020}) and for enrollment region.  Patients who were deceased or did not complete a follow-up questionnaire are omitted. \par
    \end{minipage}
\end{figure}

%%%%% F: Histogram of years of exposure %%%%%
\begin{figure}[H]
    \caption{Histogram of Revascularization Exposure by Treatment Assigned}
    \label{f:exposure histogram}
    \centering
    \begin{minipage}{16.5cm}
        {\centering{\includegraphics[width=\linewidth]{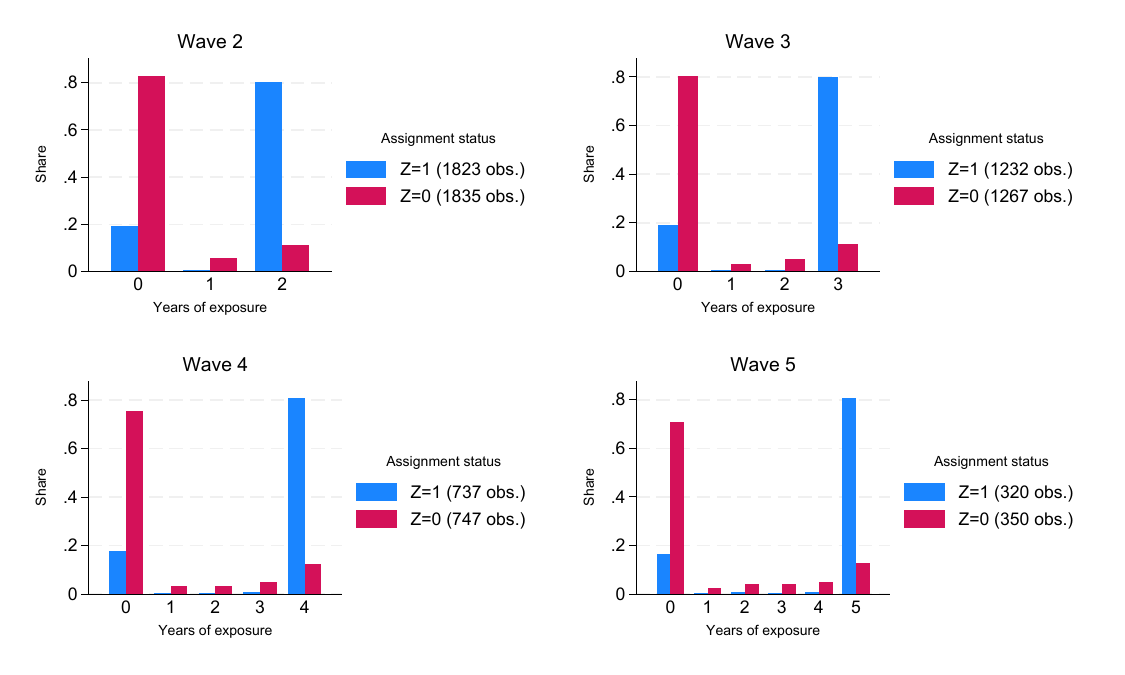}\par}}
        \footnotesize \textit{Notes:} This figure plots the histogram of revascularization exposure by follow-up wave.  Exposure $T_w=0$ for participants not revascularized as of wave $w$. \par
    \end{minipage}
\end{figure}

%%%%% F: Immediate exposure estimates %%%%%
\begin{figure}[H]
    \caption{Immediate-Exposure Revascularization Effects by Wave}
    \label{f:immediate exposure}
    \centering
    \begin{minipage}{16.5cm}
        {\centering{\includegraphics[width=\linewidth]{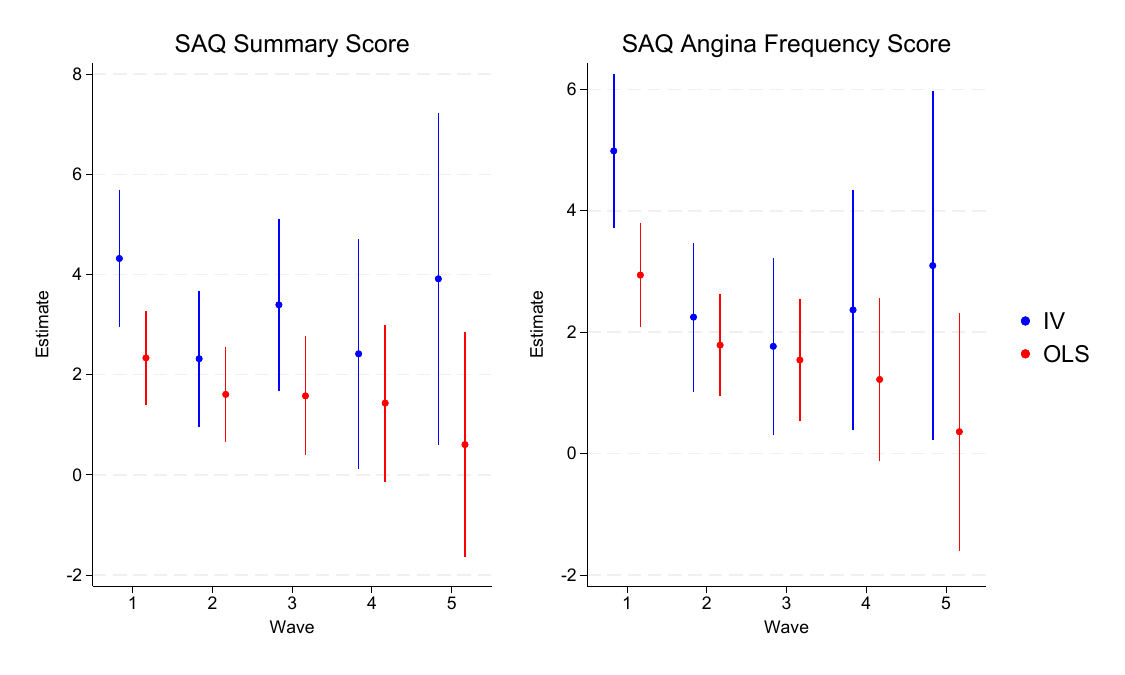}}\par}
        \footnotesize \textit{Notes:} This figure plots IV (implemented by 2SLS) and OLS estimates of average immediate exposure effects in each wave, identified by IMCO, using equation \eqref{eq:initial-esxposure_linearmodel} in Theorem \ref{th:causaleffectIMCO}. IV estimates use assignment to instrument immediate-exposure, $T_1$. OLS estimates are the corresponding (uninstrumented) as-treated $T_1$ effect. Controls are the same as in Figure \ref{f:ITTeffects}. Standard errors are clustered at the individual level. \par
    \end{minipage}
\end{figure}

%%%%% F: PO mean by group %%%%%
\begin{figure}[H]
    \caption{Potential Outcome Means for Immediate-Exposure Groups by Wave}
    \label{f:potential means}
    \centering
    \begin{minipage}{16.5cm}
        {\centering{\includegraphics[width=\linewidth]{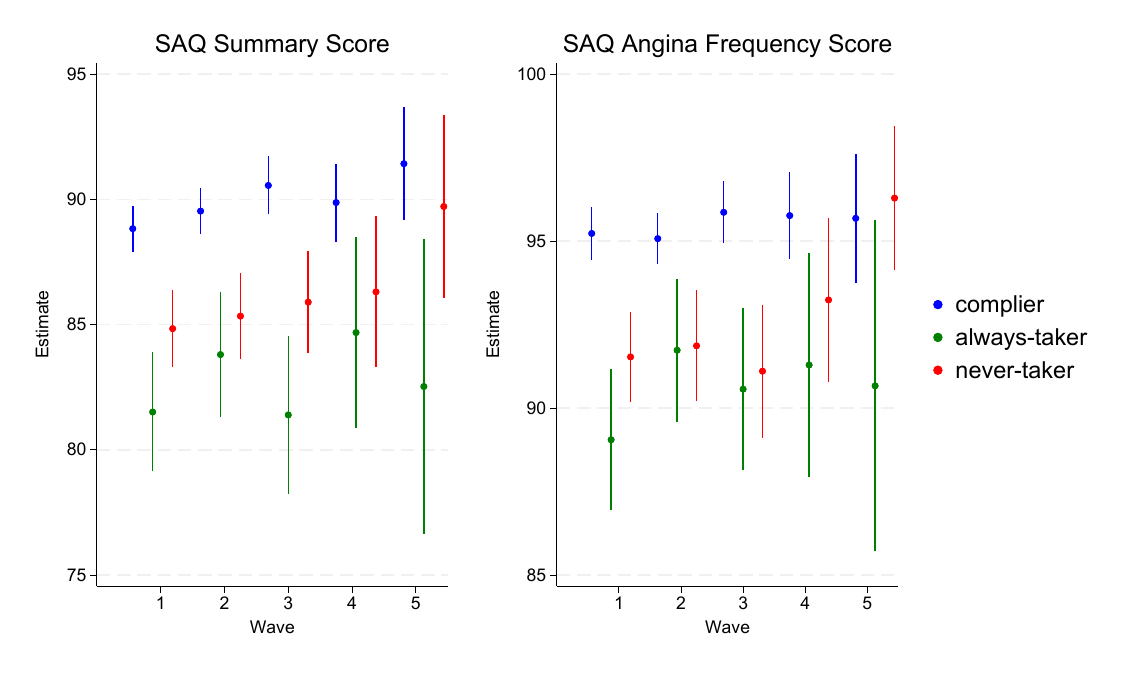}}\par}
        \footnotesize \textit{Notes:} This figure plots potential outcome means for $T_1$ compliers, always-takers, and never-takers in each wave, using equations \eqref{eq:immediate_complier_X_general},\eqref{eq:immediate_AT_X_general}, and \eqref{eq:immediate_NT_X_general} in Theorem \ref{thm:late_takers_general}. Compliers are treated immediately (i.e., in wave 1) when assigned $Z=1$ but not otherwise. Always- (never-)takers are treated (untreated) in wave 1 regardless of assignment. \par
    \end{minipage}
\end{figure}

%%%%% F Stacked %%%%%
\begin{figure}[H]
    \caption{2SLS Estimates of Revascularization Exposure Effects}
    \label{f:stacked}
    \centering
    \begin{minipage}{16.5cm}
        {\centering{\includegraphics[width=\linewidth]{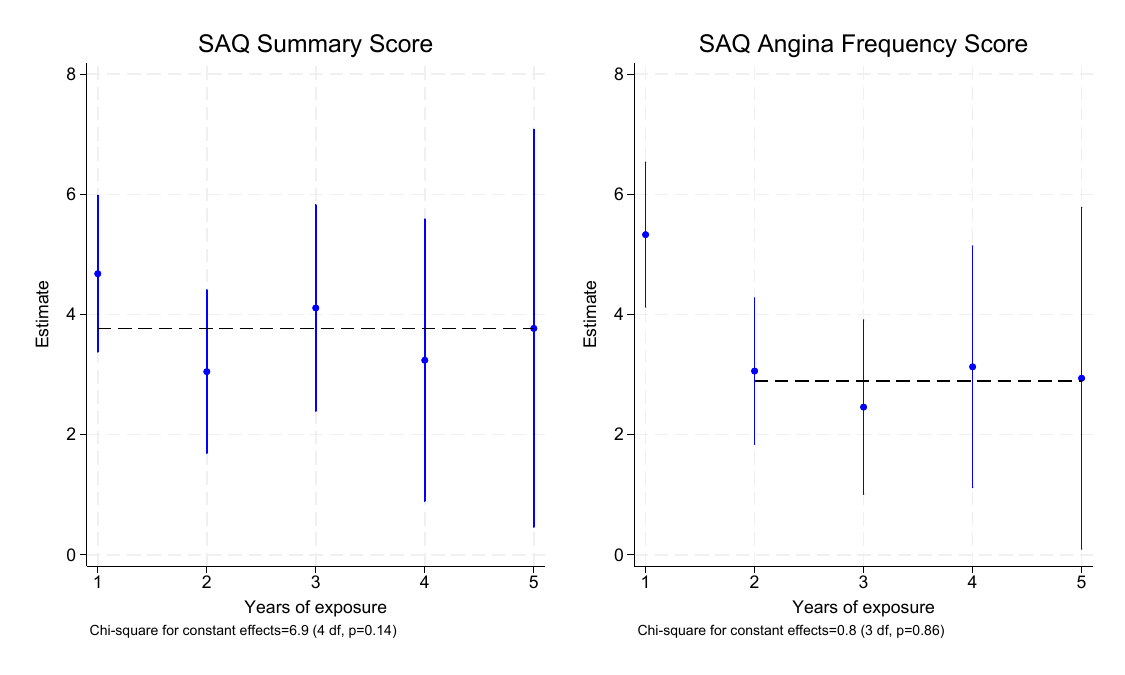}}\par}
        \footnotesize \textit{Notes:} This figure plots 2SLS estimates of average causal effects of 1-5 years of revascularization exposure, relative to never having been revascularized, using equation \eqref{eq:Lambda_IV} of Theorem \ref{corollary:Lambda}. Controls are the same as in Figure \ref{f:ITTeffects}. Standard errors are clustered at the individual level. \par
    \end{minipage}
\end{figure}

\clearpage
\begin{table}[H]
    \caption{Means, Treatment Rates, and ITT Effects}
    \label{t:sample}
        \centering
    \begin{minipage}{6.5in}
        {\centering\includegraphics[width=16cm]{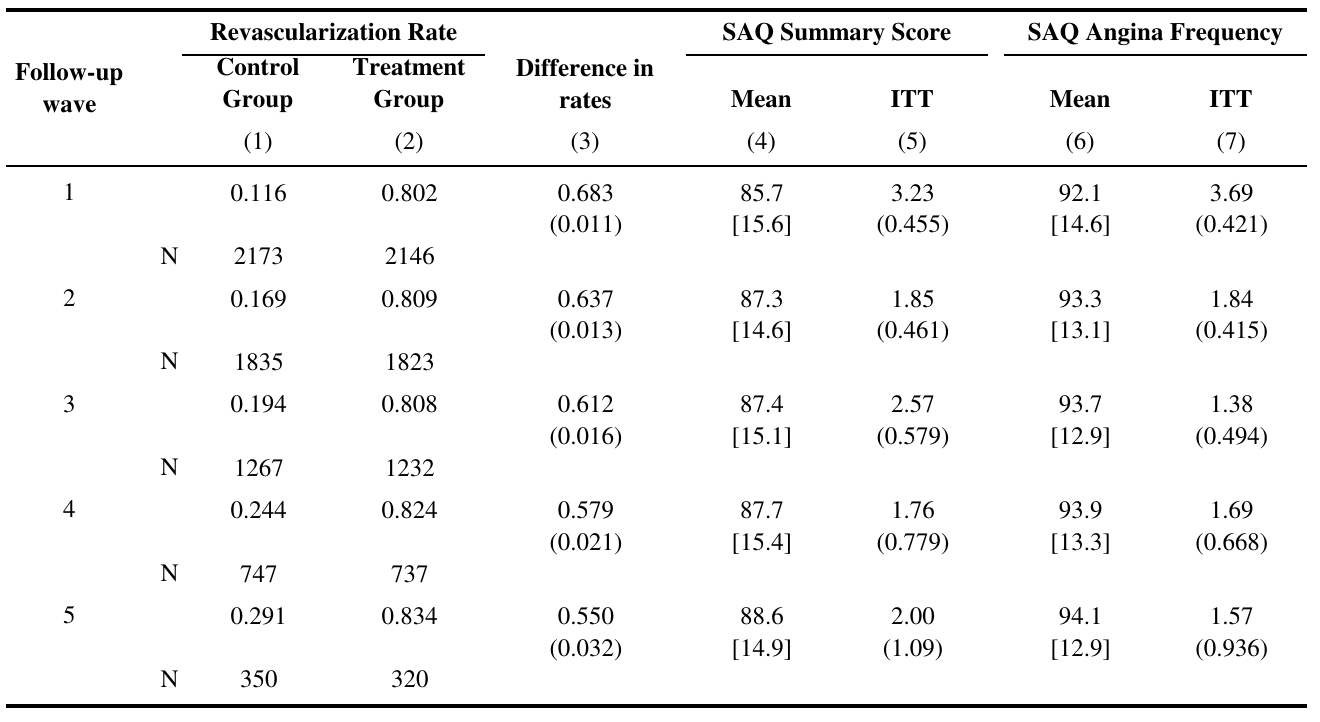}\par}
        \footnotesize \textit{Notes:} Column 1 reports the revascularization rate by wave for patients randomized to the conservative treatment group. Column 2 reports the corresponding revascularization rate for those assigned invasive. Column 3 reports the first-stage effect of treatment assignment on revascularization by wave: this is column 2 minus column 1. Columns 4 and 6 report sample means. Columns 5 and 7 report ITT estimates for the effect of treatment assigned. SAQ scores are measured at the time of follow up. Patients who were deceased or did not complete a follow-up questionnaire are omitted. Controls are the same as in Figure \ref{f:ITTeffects}. Standard deviations appear in square brackets and standard errors appear in parentheses. \par 
    \end{minipage}
\end{table}

 %%%%% Hausman test (joint)%%%%%

\begin{table}[H]
    \caption{2SLS and OLS Estimates of Revascularization Exposure Effects}
    \label{t:hausman}
        \centering
    \begin{minipage}{5.5in}
        {\centering\includegraphics[width=13.5cm]{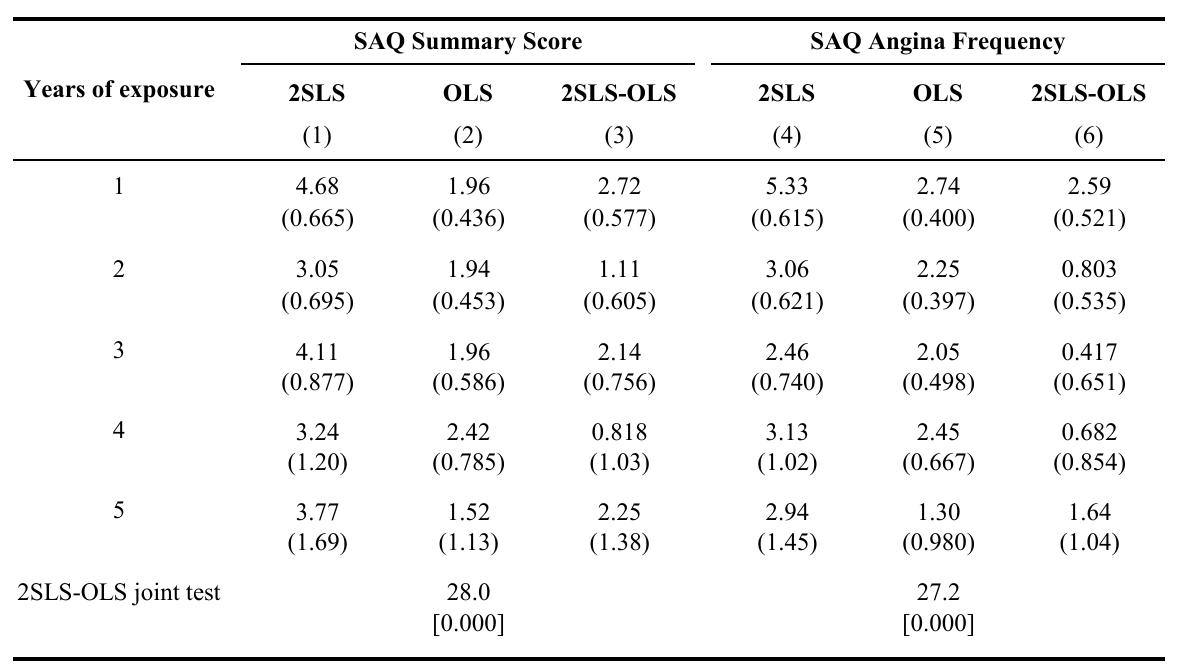}\par}
        \footnotesize \textit{Notes:} This table compares 2SLS and OLS estimates of the effects of 1-5 years of revascularization exposure computed using equation \eqref{eq:Lambda_IV}, stacking data from all waves. Columns 3 and 6 report \cite{hausman1978}-type t-tests for the difference between 2SLS and OLS estimates, where standard errors are computed using the variance of the difference of estimates. Chi-square statistics at the bottom of the table test 2SLS-OLS joint equality. This statistic has a $\chi^2(5)$ distribution under the null. Controls are the same as in Figure \ref{f:ITTeffects}. Standard errors, clustered on person, are reported in parentheses. P-values for joint tests appear in brackets in the last row.\par 
    \end{minipage}
\end{table}

% Complier Characteristics
%%%%%%%%%%%%%%%%%%%%%%%%%%%%%%%%%%%%%%%%%%%%%%%%%%%%%%
\begin{table}[H]
\caption{Complier and Always-Taker Characteristics}
\label{t:compliers}
\centering
\begin{minipage}{6.2in}
{\centering\includegraphics[width=12cm]{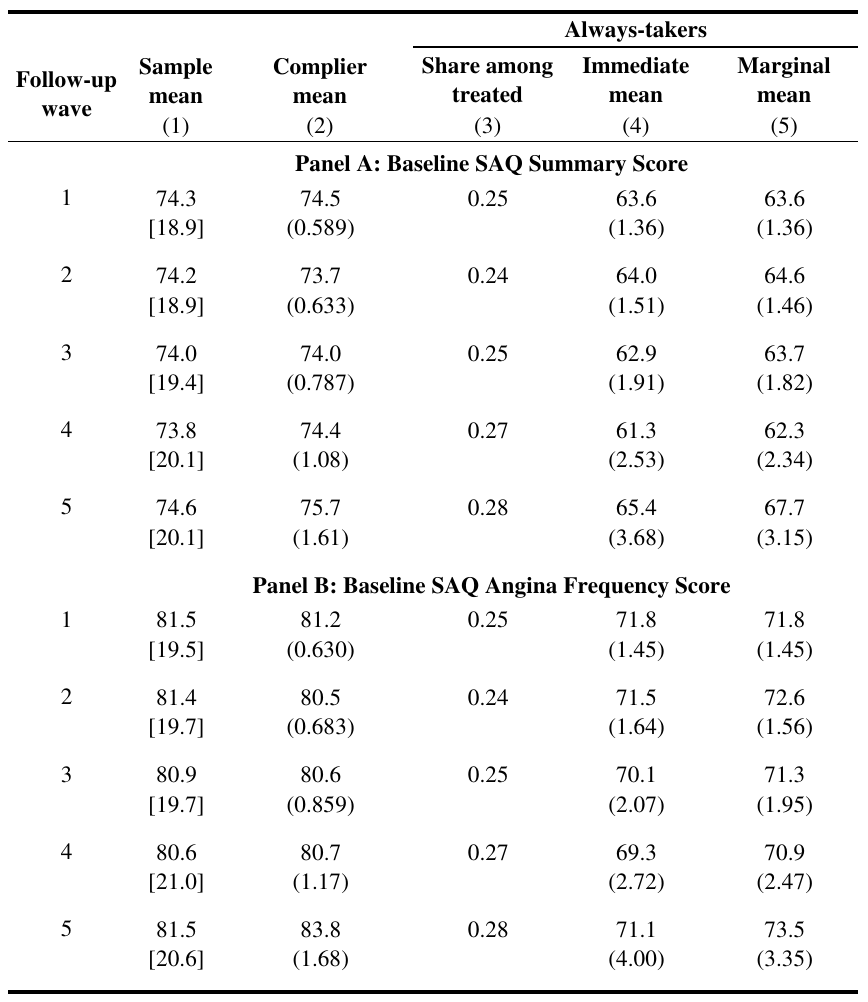}\par}
\footnotesize
\textit{Notes:} This table reports immediate complier and always-taker means obtained using Theorem \ref{thm:late_takers} for baseline summary scores (Panel A) and baseline angina frequency scores (Panel B). Column 1 shows overall sample means. Column 2 reports immediate complier means computed using equation \eqref{eq:immediate_complier_X} for each wave. The share of always-takers among the treated appears in column 3; this is the denominator of $\pi_w$ divided by sample share treated. Column 4 reports immediate always-taker means computed using equation \eqref{eq:immediate_AT_X}; column 5 reports marginal means for always-takers, computed using equation \eqref{eq:AT_X}. Standard deviations appear in square brackets and robust standard errors appear in parentheses. \par
\end{minipage}
\end{table}

\clearpage

\appendix 
\numberwithin{proposition}{section}
\numberwithin{figure}{section}
\numberwithin{table}{section}
\numberwithin{equation}{section}

\begin{center}

    {\Large \bf Appendix: ``Instrumental Variables with Time-Varying Exposure: Dynamic Effects of Revascularization on Quality of Life''}

\end{center}

\section{Proofs}

\subsection{Proof of Theorem \ref{th:causaleffectIMCO}}\label{appendix:causalfxIMCO_proof}
Note that, for any given $w\in\{1,\dots, \bar w\}$, Assumption \ref{assump:independence} implies
\begin{equation*}
\begin{split}
     & E[Y_w \mid Z=1 ] = E\left[\sum_{t=0}^w Y_w(t)\mathbf 1[T_w(1) = t]\right], \\
     & E[Y_w \mid Z=0 ] = E\left[\sum_{t=0}^w Y_w(t)\mathbf1[T_w(0) = t]\right].
\end{split}
\end{equation*}
Assumption \ref{assump:exposure_monotonicity} and IMCO then imply
\begin{equation*}
\begin{split}
    E[Y_w | Z=1 ] - E[Y_w | Z=0 ] 
    &= E\left[\sum_{0\le t< s\le w}(Y_w(s)-Y_w(t))\mathbf1[T_w(1)=s, T_w(0)=t]\right] \\
    & = E \left[ \sum_{t=0}^{w-1}(Y_w(w)-Y_w(t))\mathbf1[T_w(1)=w, T_w(0)=t] \right] \\
    & = P[T_w(1)>T_w(0)]\sum_{t=0}^{w-1}\omega_{wt} E[Y_w(w)-Y_w(t) \mid T_w(1)=w, T_w(0)=t],
\end{split}
\end{equation*}
where $\omega_{wt}\equiv P[T_w(1)=w, T_w(0)=t \mid T_w(1)>T_w(0)]$. Note that IMCO and irreversibility of treatment imply 
\begin{equation*}
    P[T_w(1)>T_w(0)] = P[T_1(1)>T_1(0)] = E[T_1 \mid Z=1] - E[T_1\mid Z=0].
\end{equation*}
That $\sum_{t=0}^{w-1}\omega_{wt}=1$ follows from the above and
\begin{equation*}
    E\left[\sum_{t=0}^{w-1} \mathbf1[T_w(1)=w, T_w(0) =t] \right]
    = E\left[\mathbf1[T_w(1)>T_w(0)] \right]
\end{equation*}
under IMCO. Moreover, for any $t<w$,
\begin{equation*}
    E[\mathbf1[T_w=t] \mid Z=1] - E[\mathbf1[T_w=t] \mid Z=0] = E[\mathbf1[T_w(1)=t] - \mathbf1[T_w(0)=t]],
\end{equation*}
which equals $-P[T_w(1)=w, T_w(0)=t]$ under monotonicity and IMCO. To conclude, note that the definition of the IV estimand in the theorem implies
\begin{equation*}
    \tau_w = \frac{E[Y_w \mid Z=1 ] - E[Y_w \mid Z=0 ]}{E[T_1 \mid Z=1 ] - E[T_1 \mid Z=0 ]}.
\end{equation*}

\subsection{Proof of Theorem \ref{thm:lambda}}\label{appendix:lambda_proof}

As in the proof of Theorem 1 in \cite{angrist_imbens95}, Assumptions \ref{assump:independence}-\ref{assump:exposure_monotonicity} can be used to show that the reduced form for wave-specific IV with outcome $Y_w$ and $T_w$ instrumented by $Z$ for given $w\in\{1,...,\Bar{w}\}$ can be written as
    \begin{equation*}
    \rho_w = \sum_{t=1}^{w} E[Y_w(t)- Y_w(t-1) | T_w(1) \ge t > T_w(0)] \pi_{wt},
    \end{equation*}  
    where $\pi_{wt}\equiv P[T_w(1) \ge t > T_w(0)]$. Under Assumption \ref{assump:wave_ignore}, this becomes:
    \begin{equation*}
    \rho_w = \sum_{t=1}^w \lambda_t\pi_{wt}. 
\end{equation*}
Let $\rho$ denote the vector of $\rho_w$, with $\Pi$ being a lower-triangular matrix with non-zero elements $\pi_{wt}$ ($w$-th row, $t$-th column). Finally, let $\lambda$ be the vector of $\lambda_t$. Then, we can write:
\begin{equation*}
\rho=\Pi\lambda.
\end{equation*}
Under Assumption \ref{assump:relevance},  $\Pi$ is invertible since it is lower-triangular with non-zero diagonal elements (absorbing treatment implies $\pi_{ww}=\pi_{11}$ for all $w$). Hence:
\begin{equation}\label{eq:2slssoln}
\lambda=\Pi^{-1}\rho.
\end{equation}
The 2SLS procedure described in the theorem using equation \eqref{eq:lambda_IV} has reduced form and first stage parameter vectors identical to the stacked wave-by-wave reduced form and first stage described here since the stacked model interacts $Z$ with wave dummies and allows wave-specific intercepts. Hence, these parameters satisfy \eqref{eq:2slssoln}.

\subsection{Proof of Theorem \ref{corollary:Lambda}}\label{appendix:Lambda_proof}
We first show that  $\sum_{i=1}^t \lambda_i = E[Y_w(t)- Y_w(0) \mid T_w(1) \ge t > T_w(0)]$ under Assumption \ref{assump:strong_wave_ignore} (strong wave ignorability), for any $w$ and $t\le w$. Note that for all $i < t$:
\begin{align*}
    \{T_{i}(1) \ge i > T_{i}(0)\} \iff \{T_t(1) \ge t > T_t(0)\},
\end{align*}
since, with an absorbing treatment, those with $t$ months of exposure at wave $t$ must have been exposed initially. 
Hence, for any $i<t$: 
\begin{align*}
    E[Y_{i}(i) - Y_{i}(i-1) | T_{i}(1) \ge i > T_{i}(0)]  = E[Y_i(i)-Y_i(i-1) | T_t(1) \ge t > T_t(0)].
\end{align*}
Moreover, for any $i<t$, strong wave ignorability implies:
\begin{align*}
    &E[Y_{i}(i) - Y_i(i-1) | T_i(1) \ge i > T_i(0)] = E[Y_w(i) - Y_w(i-1) | T_w(1) \ge i > T_w(0)].
\end{align*}
and
\begin{align*}
    &E[Y_i(i) - Y_i(i-1) | T_t(1) \ge t > T_t(0)] = E[Y_w(i) - Y_w(i-1) | T_w(1) \ge t > T_w(0)].
\end{align*}
Thus:
\begin{align*}
    \sum_{i=1}^t \lambda_i & = \sum_{i=1}^t E[Y_w(i) - Y_w(i-1) | T_w(1) \ge i > T_w(0)]\\
    & = \sum_{i=1}^t E[Y_w(i) - Y_w(i-1) | T_w(1) \ge t > T_w(0)] \\
    & = E\left[\sum_{i=1}^tY_w(i) - Y_w(i-1) \middle| T_w(1) \ge t > T_w(0)\right] \\
    &=E[Y_w(t) - Y_w(0) \mid  T_w(1) \ge t > T_w(0)].
\end{align*}

For the second part of the theorem, note that $R_{wt}=\sum_{i=1}^{t}D_{wi}$. Thus, the 2SLS estimand instrumenting $R_{wt}$ in equation \eqref{eq:Lambda_IV} generates the sum of coefficients from the 2SLS estimand instrumenting $D_{wt}$ in equation \eqref{eq:lambda_IV}.  That is,  $\sum_{i=1}^{t}\lambda_i=\Lambda_t$.

\begin{comment}
Additionally, note that for any $v\in\{2,...,w\}$,
\begin{equation*}
\begin{split}
E[Y_w(v)-Y_w(0) \mid T_w(1)\ge1>T_w(0)]
&=
\sum_{t=1}^v E[Y_w(t)-Y_w(t-1) \mid T_w(1)\ge1>T_w(0)]\\
&=
E[Y_w(1)-Y_w(0) \mid T_w(1)\ge1>T_w(0)],
\end{split}
\end{equation*}
where the second equality follows from setting $t^\prime=1$ in the condition on incremental effects on the statement of the theorem. Therefore, $\tau_w = E[Y_w(v)-Y_w(0) \mid T_w(1)\ge1>T_w(0)]$ for any $v\in\{1,...,w\}$.
\end{comment}

\subsection{Marginal Potential Outcome Means}\label{appendix:always_taker_mean_general}

\renewcommand{\thetheorem}{\Alph{section}\arabic{theorem}}
\setcounter{theorem}{0} 

This appendix extends Theorem \ref{thm:late_takers} to identify expectations of potential outcomes for different types of compliers and always-takers. Theorem \ref{thm:late_takers} in the text, which establishes identification of baseline covariate means in our dynamic setup, is a consequence of this more general result.  In particular, Theorem \ref{thm:late_takers} follows from the theorem below by setting $Y_w(t)=X$ for all $w$ and $t$.

\begin{theorem}\label{thm:late_takers_general}
    Suppose Assumptions \ref{assump:independence}-\ref{assump:relevance} hold. Then: \
\begin{itemize} 
    \item[i.] \emph{Immediate complier means} are given by
    \begin{align}
    \label{eq:immediate_complier_X_general}
    E[Y_w(w) \mid T_1(1)>T_1(0)] 
    =\frac{E[\mathbf{1}[T_1=1]\times Y_w\mid Z=1] - E[\mathbf{1}[T_1=1]\times Y_w\mid Z=0]}{E[\mathbf{1}[T_1=1]\mid Z=1] - E[\mathbf{1}[T_1=1]\mid Z=0]}.
    \end{align}
    \item[ii.] \emph{Immediate always-taker means} are given by
    \begin{equation}
    \label{eq:immediate_AT_X_general}
    E[Y_w(w) \mid T_1(1)=T_1(0)=1]=E[ Y_w \mid T_1=1, Z=0].
    \end{equation}
\end{itemize}
\noindent Moreover, if Assumption \ref{assump:imco} also holds, we have that:
\begin{itemize}
\setcounter{enumi}{2}
\item[iii.] Complier means equal immediate complier means in each wave:
    \begin{equation}
    \label{eq:latercomplier_ag_X_general}
    E[Y_w(w) \mid T_w(1)>T_w(0)]=E[Y_w(w)\mid T_1(1)>T_1(0)]
    \end{equation}
    for all $w\in\{1,...,\Bar{w}\}$;
    \item[iv.] For each $w>1$ and $t\in\{0,\dots,w-1\}$ such that
    \begin{equation*}
        E[\mathbf{1}[T_w=t]\mid Z=1]-E[\mathbf{1}[T_w=t]\mid Z=0]\neq 0,
    \end{equation*}
    \emph{disaggregated complier means} are given by
    \begin{equation}
    \label{eq:latercomplier_dsg_X_general}
    E[Y_w(t) \mid T_w(1)=w, T_w(0)=t]=
    \frac{E[\mathbf{1}[T_w=t]\times Y_w\mid Z=1]-E[\mathbf{1}[T_w=t]\times Y_w\mid Z=0]}{E[\mathbf{1}[T_w=t]\mid Z=1]-E[\mathbf{1}[T_w=t]\mid Z=0]};
    \end{equation}
    \item[v.] For each $w>1$, \emph{later always-taker means} are given by
    \begin{equation}
    \label{eq:laterATs_X_general}
    E[Y_w(T_w(1))\mid w>T_w(1)=T_w(0)\ge 1]
    = E[Y_w\mid w>T_w\ge1, Z=1].
    \end{equation}
    \emph{Marginal always-takers means}, which average immediate and later always-takers, can be obtained using
    \begin{equation}
    \label{eq:AT_X_general}
    \begin{split}
    & E[Y_w(T_w(1))\mid T_w(1)=T_w(0)\ge 1] \\
    &= \pi_w E[ Y_w \mid T_w=w, Z=0]
    + (1-\pi_w)E[Y_w\mid w>T_w\ge1, Z=1],
    \end{split}
    \end{equation}
    where $\pi_w\equiv\frac{E[\mathbf{1}[T_w=w]\mid Z=0]}{E[\mathbf{1}[T_w=w]\mid Z=0] + E[\mathbf{1}[1\le T_w<w]\mid Z=1]}$.
    \item[vi.] For each each $t\in\{0,\dots,w-1\}$, \emph{disaggregated always- and never-taker means} are given by
    \begin{equation}
    \label{eq:laterATs_dsg_X_general}
    E[Y_w(t)\mid T_w(1)=T_w(0)=t]=
    E[ Y_w\mid T_w=t, Z=1].
    \end{equation}
    For never-takers, set $t=0$. Identification of never-taker means does not require IMCO (Assumption \ref{assump:imco}). Finally, for any $w\in\{1,\ldots,\bar w\}$, \emph{immediate never-taker means} are
    \begin{equation}
    \label{eq:immediate_NT_X_general}
    E[Y_w \mid T_1=0, Z=1] = E[Y_w \mid 0\le T_w<w, Z=1] = \sum_{t=0}^{w-1}\tilde\omega_{wt} E[Y_w(t) \mid T_w(1)=T_w(0)=t],
    \end{equation}
    where $\tilde\omega_{wt}\equiv P[T_w(1)=T_w(0)=t\mid 0\le T_w<w]$ for all $t\in\{0,\ldots, w-1\}$.
\end{itemize}
\end{theorem}
\begin{proof}
Consider the first two results, \eqref{eq:immediate_complier_X_general} and \eqref{eq:immediate_AT_X_general}, which do not require Assumption \ref{assump:imco}. Fix $w\in\{1,...,\Bar{w}\}$. Because treatment is irreversible, $T_w=w$ if, and only if, $T_1=1$. Monotonicity then implies:
\begin{equation*}
\begin{split}
E[\mathbf{1}[T_1=1]\times Y_w\mid Z=1]
=&
E[\mathbf{1}[T_1(1)=1]\times Y_w(w)]\\
=&
P[T_1(1)=T_1(0)=1]E[ Y_w(w)\mid T_1(1)=T_1(0)=1]\\
&+
P[T_1(1)>T_1(0)]E[ Y_w(w)\mid T_1(1)>T_1(0)].
\end{split}
\end{equation*}
Analogously, $E[\mathbf{1}[T_1=1]\times Y_w\mid Z=0]=P[T_1(1)=T_1(0)=1]E[ Y_w(w)\mid T_1(1)=T_1(0)=1]$, which establishes equation \eqref{eq:immediate_complier_X_general}. Equation \eqref{eq:immediate_AT_X_general} follows from $E[Y_w\mid T_1=1, Z=0]=E[Y_w(w)\mid T_1(0)=1]=E[Y_w(w)\mid T_1(1)=T_1(0)=1]$ under monotonicity. 

Now, for a given $w\in\{2,...,\Bar{w}\}$, consider results in the theorem that depend on Assumption \ref{assump:imco}:
\begin{itemize}
     \item[1.] To establish \eqref{eq:latercomplier_ag_X_general}, note that because treatment is irreversible, $T_1(1)>T_1(0)$ implies $T_w(1)>T_w(0)$. Conversely, $T_w(1)>T_w(0)$ implies $T_w(1)=w$ almost surely, which in turn implies $T_1(1)>T_1(0)$ because treatment is irreversible.
     
     \item[2.] Fix $t\in\{0,...,w-1\}$. Note that when $Z=1$, $\mathbf{1}[T_w=t]=1$ if and only if $T_w(1)=t$. Therefore, under Assumptions \ref{assump:exposure_monotonicity} and \ref{assump:imco}, $E[\mathbf{1}[T_w=t]\mid Z=1]=P[T_w(1)=T_w(0)=t\mid Z=1]$ since compliers have $T_w(1)=w$ almost surely. Moreover, under monotonicity, 
    \begin{equation}
    \label{eq:laterAT_dsg_proof}
    E[\mathbf{1}[T_w=t]\times Y_w\mid Z=1] = P[T_w(1)=T_w(0)=t]E[Y_w(t)\mid T_w(1)=T_w(0)=t].
    \end{equation}
    When $Z=0$, in addition to always-takers with $T_w(0)=t$, $\mathbf{1}[T_w=t]=1$ for compliers with $T_w(0)=t<w=T_w(1)$. Thus,
    \begin{equation*}
    \begin{split}
    E[\mathbf{1}[T_w=t]\times Y_w\mid Z=0] =&
    P[T_w(1)=T_w(0)=t]E[Y_w(t)\mid T_w(1)=T_w(0)=t] \\
    &+
    P[T_w(1)=w, T_w(0)=t]E[Y_w(t)\mid T_w(1)=w, T_w(0)=t],
    \end{split}
    \end{equation*}
    which together with \eqref{eq:laterAT_dsg_proof} establishes  \eqref{eq:latercomplier_dsg_X_general}. Equation \eqref{eq:laterATs_dsg_X_general} follows from the same argument used to establish \eqref{eq:laterAT_dsg_proof}.
    
    \item[3.] To show \eqref{eq:laterATs_X_general}, note that monotonicity implies
    \begin{equation*}
    \begin{split}
    E[Y_w \mid w>T_w\ge1, Z=1]
    &=
    E[Y_w(T_w(1)) \mid w>T_w(1)\ge1]\\
    &=
    E[Y_w(T_w(1)) \mid w>T_w(1)=T_w(0)\ge 1],
    \end{split}
    \end{equation*}
    since $P[T_w(1)=T_w(0) \mid 1\le T_w(1)<w ]=1$ under Assumption \ref{assump:imco}.

    \item[4.] To show \eqref{eq:AT_X_general}, note that 
    \begin{equation*}
    \begin{split}
    &E[Y_w(T_w(1))\mid  T_w(1)=T_w(0)\ge 1] \\
    =& P[T_w(1)=T_w(0)=w\mid T_w(1)=T_w(0)\ge 1]E[ Y_w(T_w(1)) \mid T_w(1)=T_w(0)=w] \\
    &+
    P[w>T_w(1)=T_w(0)\ge1\mid T_w(1)=T_w(0)\ge 1]E[ Y_w(T_w(1)) \mid w>T_w(1)=T_w(0)\ge1].
    \end{split}
    \end{equation*}
    The fact that $E[ Y_w(T_w(1)) \mid T_w(1)=T_w(0)=w]=E[Y_w\mid T_w=w, Z=0]$ follows from equation \eqref{eq:immediate_AT_X_general} by noting that $T_1(1)=T_1(0)=1$ if, and only if, $T_w(1)=T_w(0)=w$ because treatment is irreversible. The result that $E[ Y_w(T_w(1)) \mid w>T_w(1)=T_w(0)\ge1]=E[X_w \mid w>T_w\ge1, Z=1]$ follows from \eqref{eq:laterATs_X_general}. Finally, monotonicity implies $E[\mathbf{1}[T_w=w]\mid Z=0]=P[T_w(1)=T_w(0)=w]$ and, under IMCO, $E[\mathbf{1}[1\le T_w<w]\mid Z=1]=P[w>T_w(1)=T_w(0)\ge1]$. Thus, $\pi_w=P[T_w(1)=T_w(0)=w\mid T_w(1)=T_w(0)\ge 1]$.

    Note that equation \eqref{eq:AT_X} uses $E[X\mid T_1=1, Z=0]$ while equation \eqref{eq:AT_X_general} uses $E[Y_w\mid T_w=w, Z=0]$. Because treatment is irreversible, conditioning on $T_1=1$ is equivalent to condition on $T_w=w$. In the case of equation \eqref{eq:AT_X}, because $X$ does not vary across waves, this means that $E[X\mid T_1=1, Z=0]$ can be computed only once even though $E[X\mid T_w(1)=T_w(0)\ge1]$ varies with $w$. On the other hand, in equation \eqref{eq:AT_X_general}, $E[Y_w\mid T_w=w, Z=0]$ varies with $w$ even though the latent group does not. 

    \item[5.] To conclude, fix $w\in\{1,\ldots,\bar w\}$. The first equality in \eqref{eq:immediate_NT_X_general} is immediate from irreversibility of treatment. The second equality then follows from 
    \begin{equation*}
        E[Y_w\mid 0\le T_w<w, Z=1]
        =
        \sum_{t=0}^{w-1}P[T_w=t\mid 0\le T_w<w, Z=1]E[Y_w \mid T_w=t, Z=1].
    \end{equation*}
    From \eqref{eq:laterATs_dsg_X_general}, $E[Y_w \mid T_w=t, Z=1]$ equals $E[Y_w(t)\mid T_w(1)=T_w(0)=t]$ for all $t\in\{0,\ldots, w-1\}$. Lastly, for any $t\in\{0,\ldots, w-1\}$, Assumption \ref{assump:independence} implies
    \begin{equation*}
        P[T_w=t\mid 0\le T_w<w, Z=1] = P[T_w(1) = t\mid 0\le T_w<w],
    \end{equation*}
    which equals $P[T_w(1)=T_w(0)=t\mid 0\le T_w<w]$ under IMCO.
    
\end{itemize}
\end{proof}

\section{Supplemental Material}

\subsection{Comparison to \cite{Bowden2025}}
\label{appendix: bowden}
This appendix compares our identification strategy to the use of instruments in Section 2.3 of \cite{Bowden2025}.
\cite{Bowden2025} is more restrictive in that it does not allow for control-group crossovers; in our notation, this means $T_w(0)=0$ with probability one. A second difference is that \cite{Bowden2025} impose a parametric specification for the dynamic incremental effects $\lambda_t$. In our notation, this amounts to imposing the restriction: 
\begin{equation}
\label{eq:bowden_geometric}
    \lambda_t = \beta \alpha^t,
\end{equation}
for $\beta,\alpha\in\mathbb{R}$, with $\alpha$ implicitly assumed to be non-negative (see equation (6) and related discussions in \cite{Bowden2025}).\footnote{Note that the definition of the causal effects $\beta_k(j)$ in their Section 2.1 restricts potential outcome heterogeneity across IV latent types.} This parametric restriction on the incremental treatment effects implies our wave ignorability assumption (Assumption \ref{assump:wave_ignore}) but is much stronger: for example, with $\alpha\ge 0$, it doesn't allow for non-monotonic dynamics of incremental effects or for cumulative treatment effects that are initially positive and then decrease over time. Even with unrestricted $\alpha$, Equation \eqref{eq:bowden_geometric} rules out several plausible paths for treatment effects---such as when effects persist for more than one period before decaying (see Appendix Figure \ref{f:bowden}).

We now extend the \cite{Bowden2025} framework to allow for control-group crossovers and show how their parametric model for incremental effects can be imposed and tested in our setting. Using the notation of Theorem \ref{thm:lambda}, the estimator \cite{Bowden2025} propose is a generalized method of moments (GMM) estimator of $\beta$ and $\alpha$ based on the moment condition:
\begin{equation}
\label{eq:bowden_momentconditions}
    E\left[(Z - E[Z])
      \begin{bmatrix}  Y_1 - \beta D_{11} Z \\
      Y_2 - (\beta\alpha D_{22} + \beta D_{21})Z \\
      \vdots \\
      Y_{\bar w} - \sum_{j=1}^{\bar w}D_{wj}\beta\alpha^{j - 1}Z
      \end{bmatrix}\right] = 0.
\end{equation}
We propose a modified estimator which replaces each $D_{wt}Z$ term with just $D_{wt}$:
\begin{equation}
\label{eq:bowden_ischemia}
    E\left[(Z - E[Z])
      \begin{bmatrix}  Y_1 - \beta D_{11} \\
      Y_2 - \beta\alpha D_{22} - \beta D_{21} \\
      \vdots \\
      Y_{\bar w} - \sum_{j=1}^{\bar w}D_{wj}\beta\alpha^{j - 1}
      \end{bmatrix}\right]= 0.
\end{equation}
This modification is inconsequential in settings without control-group crossovers: $Z=0$ implies $D_{wt}=0$ in that case, such that $D_{wt}Z=D_{wt}$. In the more general setting we consider, an estimator based on equation \eqref{eq:bowden_ischemia} can be seen as a restricted version of the 2SLS estimator in Theorem \ref{thm:lambda} which imposes the parametric restriction \eqref{eq:bowden_geometric} on incremental effects. That is, if one were to relax the equation \eqref{eq:bowden_geometric} assumption by replacing $\beta\alpha^t$ with unrestricted $\lambda_t$ coefficients in equation \eqref{eq:bowden_ischemia}, the resulting estimates would coincide exactly with the 2SLS estimates in Theorem \ref{thm:lambda}.\footnote{Since $Z$ is binary, the moment conditions $E[\eta_w]=E[Z\eta_w]=0$, $w=1,\dots,\bar w,$ define the 2SLS estimand in Theorem \ref{thm:lambda}. For wave $w$, these imply $\mu + \alpha_w = E[Y_w-\sum_{t=1}^w\lambda_t D_{wt}]$ and $E[(Z-E[Z])\eta_w]=0$. Therefore, they imply $E[(Z-E[Z])(\eta_w - \mu - \alpha_w)]=0$; stacking across waves gives \eqref{eq:bowden_ischemia}. The relevance condition (Assumption \ref{assump:relevance}) implies that the solutions to the 2SLS moment condition and \eqref{eq:bowden_ischemia} are unique and hence coincide.}

Table \ref{t:bowden} reports GMM estimates of $\beta$ and $\alpha$ for the two primary ISCHEMIA outcomes, using the extended moment condition \eqref{eq:bowden_ischemia} after adding linear controls as in the main analysis. Since $\beta$ and $\alpha$ are overidentified in this setting, we use the efficiently-weighted two-step GMM estimator and report overidentification test statistics and p-values \citep{hansenGMM}. Columns 1-3 of the table estimate a version of the model that is faithful to the original \cite{Bowden2025} specification, with $\alpha\ge 0$, while columns 4-6 allow $\alpha$ to be negative.\footnote{In practice, to impose $\alpha\ge 0$, we first minimize the GMM criterion function allowing $\alpha$ to take only strictly positive values, which drives it to zero. Then, we impose $\alpha=0$ and choose $\beta$ to minimize the GMM criterion under such constraint. Columns 1-3 report results from this second step.} Columns 1 and 4 pool the two ISCHEMIA outcomes in order to increase precision.

The table shows that the augmented \cite{Bowden2025} model of incremental effects does not fit the ISCHEMIA data well under natural parameter values. In columns 1-3 the estimate of $\alpha$ is driven to the lower bound of zero and overidentification tests largely reject---with $p<0.01$ in the pooled model and the SAQ Angina Frequency outcome. Even when allowing negative $\alpha$ in columns 4-6, the pooled model yields a marginal overidentification test rejection. Broadly, these results highlight the value of our nonparametric approach.

\subsection{Latent-Index Representation of Monotonicity}\label{appendix:vytlacil}

This appendix shows that Assumption \ref{assump:exposure_monotonicity} is implied by a nonparametric latent-index model. The argument here is similar to that in \citet{vytlacil2002} for a single static treatment. Let
\[
S_i(z)=\inf\{s\in\{1,\ldots,\bar w\}:T_{is}(z)>0\},
\]
with $\inf\emptyset=\infty$. Since revascularization is absorbing,
\[
T_{iw}(z)=\sum_{s=1}^w \mathbf{1}\{S_i(z)\le s\}.
\]

\begin{proposition} The following two statements
are equivalent:
\[
T_{iw}(1)\ge T_{iw}(0)
\qquad\text{for every } w=1,\ldots,\bar w,
\]
and
\[
S_i(1)\le S_i(0).
\]
Moreover, this condition is equivalent to the existence of a scalar
assignment index \(\nu(z)\), normalized as \(\nu(0)=0\) and \(\nu(1)=1\),
and patient-specific random thresholds
\[
C_{i1}\ge C_{i2}\ge\cdots\ge C_{i\bar w}
\]
such that
\[
\mathbf{1}\{S_i(z)\le s\}
=
\mathbf{1}\{\nu(z)\ge C_{is}\},
\qquad s=1,\ldots,\bar w,\quad z\in\{0,1\}.
\]
Equivalently,
\[
T_{iw}(z)
=
\sum_{s=1}^w \mathbf{1}\{\nu(z)\ge C_{is}\}.
\]
\end{proposition}

\begin{proof}
First, suppose \(S_i(1)\le S_i(0)\). Then, for every \(s\),
\[
\mathbf{1}\{S_i(1)\le s\}
\ge
\mathbf{1}\{S_i(0)\le s\}.
\]
Summing from \(s=1\) to \(w\) gives
\[
T_{iw}(1)
=
\sum_{s=1}^w \mathbf{1}\{S_i(1)\le s\}
\ge
\sum_{s=1}^w \mathbf{1}\{S_i(0)\le s\}
=
T_{iw}(0).
\]
Conversely, suppose \(S_i(1)>S_i(0)\). At \(w=S_i(0)\), we have
\(T_{iw}(0)>0\) but \(T_{iw}(1)=0\), contradicting
\(T_{iw}(1)\ge T_{iw}(0)\). Therefore
\[
T_{iw}(1)\ge T_{iw}(0)\ \text{for all }w
\qquad\Longleftrightarrow\qquad
S_i(1)\le S_i(0).
\]

It remains to establish the latent-index representation. If the
latent-index representation holds, monotonicity follows immediately because
\(\nu(1)\ge \nu(0)\):
\[
\mathbf{1}\{\nu(1)\ge C_{is}\}
\ge
\mathbf{1}\{\nu(0)\ge C_{is}\}
\]
for every \(s\). Summing over \(s\le w\) yields
\(T_{iw}(1)\ge T_{iw}(0)\).

For the converse, assume monotonicity, or equivalently \(S_i(1)\le S_i(0)\).
Normalize \(\nu(0)=0\) and \(\nu(1)=1\). Define
\[
C_{is}=
\begin{cases}
2, & s<S_i(1),\\
1/2, & S_i(1)\le s<S_i(0),\\
-1, & s\ge S_i(0).
\end{cases}
\]
These thresholds are weakly decreasing in \(s\). Moreover,
\[
\mathbf{1}\{1\ge C_{is}\}
=
\mathbf{1}\{S_i(1)\le s\},
\]
and
\[
\mathbf{1}\{0\ge C_{is}\}
=
\mathbf{1}\{S_i(0)\le s\}.
\]
Thus,
\[
\mathbf{1}\{S_i(z)\le s\}
=
\mathbf{1}\{\nu(z)\ge C_{is}\},
\qquad z\in\{0,1\}.
\]
This establishes the result.
\end{proof}

\subsection{Semiparametric Efficiency}\label{sec:efficiency}

This appendix shows that the sample-analog 2SLS estimator is semiparametrically efficient for the 2SLS estimand in Theorem \ref{thm:lambda}. Proofs of analogous results for Theorems \ref{th:causaleffectIMCO} and \ref{corollary:Lambda} follow by the same steps. The results stem from the fact that the estimands in Theorem \ref{th:causaleffectIMCO}-\ref{corollary:Lambda} are defined by moment conditions whose number equals the number of parameters, resulting in a just-identified Generalized Method of Moments (GMM) model for the distribution of data. \cite{chensantos_overid} show that when the Jacobian matrix in GMM  estimation is full-rank the notion of just-identification in GMM coincides, under mild conditions, with a broader notion of just-identification that implies all regular estimators are asymptotically equivalent (see their GMM Illustration Example and Theorem 3.1). What follows specializes this more general efficiency result to our setting.

 Let $X$ be the random vector stacking instrument, outcomes, and treatment status for all waves:  $$X\equiv(Z,Y_{1},\ldots,Y_{\bar w},T_{1},\ldots,T_{\bar w})^\prime.$$ For simplicity, when dealing with derivatives of functionals over probability measures, we assume that the distribution of $Y\equiv (Y_1,\ldots, Y_{\bar w})^\prime$ has bounded support (which also guarantees 2SLS is regular, per \cite{chensantos_overid}, Remark 3.1). This and the other support restrictions for the instrument (binary) and treatment (integer-valued and irreversible) are encoded in the convention that $X$ takes values in a set $\mathcal X$. We then let $\mathcal P$ be the set of probability measures over $\mathcal X$ such that $0<E_P[Z]<1$ and $E_P[T_1\mid Z=1]-E_P[T_1\mid Z=0]>0$, where, in this Appendix, we index expectations with the probability measure. 

For $w\in\{1,\ldots,\bar w\}$ and $t\le w$, define
$D_{wt}\equiv\mathbf 1[T_w\ge t]$. Let $\mathbf D$ denote the random
$\bar w\times \bar w$ lower-triangular  matrix where entry $(w,t)$ is equal to $D_{wt}$ for $t\neq w$ and zero for $t>w$. Finally, for a probability measure $P\in\mathcal P$, let
$$\rho(P) \equiv E_P[Y\mid Z=1]-E_P[Y\mid Z=0],
\qquad
\Pi(P) \equiv E_P[\mathbf D\mid Z=1]-E_P[\mathbf D\mid Z=0].$$
Since $\Pi(P)$ is lower triangular with a non-zero diagonal, it is invertible for any $P\in\mathcal P$. By Appendix \ref{appendix:lambda_proof}, the functional $$P\mapsto \lambda(P) \equiv \Pi(P)^{-1}\rho(P)$$ identifies the target parameter in Theorem \ref{thm:lambda}. 

We consider efficient estimation of this functional when defined on $\mathcal P$. The next proposition establishes that the asymptotic variance of any regular estimator of $P\mapsto \lambda(P)$ is no smaller than that of the 2SLS estimator.

\begin{proposition}
    Consider the set of probability measures $\mathcal P$ and the functional $\lambda:\mathcal P\to \mathbb{R}^{\bar w}$ defined above. Then the asymptotic variance matrix of any regular sequence of estimators of $\lambda$ (as defined in \cite{van2000asymptotic}, Chapter 25)  is bounded from below, in the positive semi-definite order, by the asymptotic variance matrix of the 2SLS estimator.
\end{proposition}
\begin{proof}
    Fix $P\in\mathcal P$. The tangent space of $\mathcal P$ at $P$ is the set of real-valued functions that are mean-zero and square-integrable with respect to $P$, $L_0^2(P)$. To see this, proceed as in Example 25.16 of \cite{van2000asymptotic}: fix $g\in L^2_0(P)$ to construct a one-dimensional parametric submodel $(P_t)_{t\in[0,1]}$ such that 
    \begin{equation*}
    t\mapsto \frac{dP_t}{dP} = \frac{
    \psi(t g)
    }{
    \int \psi(t g(x))dP(x)
    },
    \quad 
    \mathbb{R}\ni u\mapsto \psi(u)\equiv \frac{2}{1+\exp(-2u)},
    \end{equation*}
    where $dP_t/dP$ is the Radon-Nykodim derivative of $P_t$ with respect to $P$. This parametric submodel is in $\mathcal P$ for sufficiently small $t\in[0,1]$: since $u\mapsto\psi(u)$ is bounded, the Dominated Convergence Theorem implies that $\int|dP_t - dP|\to0$ as $t\downarrow0$ and, therefore, that 
    \begin{equation*}
        \big|E_{P_t}[\mathbf1[Z=z]] - E_P[\mathbf1[Z=z]]\big|\to0, \quad \big|E_{P_t}[\mathbf1[Z=z]\, T_1]- E_P[\mathbf1[Z=z]\, T_1]\big|\to0
    \end{equation*}
    for $z\in\{0,1\}$ as $t\downarrow0$. Moreover, the submodel $(P_t)_{t\in[0,1]}$ is differentiable in quadratic mean at $P=P_0$ (i.e., it satisfies (25.13) in \cite{van2000asymptotic}) as per Example 25.16 in \cite{van2000asymptotic}. This shows that the tangent space of $\mathcal P$ at $P$ is indeed $L^2_0(P)$.

    Now fix an arbitrary one-dimensional parametric submodel $(P_t)_{t\in[0,1]}$ in $\mathcal P$ such that $P_0=P$ and $(P_t)_{t\in[0,1]}$ is differentiable in quadratic mean at $P$ with score function $g\in L^2_0(P)$. Differentiability in quadratic mean and the Cauchy-Schwarz inequality imply 
    \begin{equation*}
    \begin{split}
        \left\Vert \frac{dP_t - dP}{t}-g\,dP \right\Vert_{1}
        \le &
        \left\Vert \frac{g}{2}\,dP^{1/2}\right\Vert_{2}  
        \left\Vert dP_t^{1/2} - dP^{1/2}\right\Vert_{2}\\
        &+
        \left\Vert \frac{dP_t^{1/2} - dP^{1/2}}{t} -\frac{g}{2}\,dP^{1/2}\right\Vert_{2}  
        \left\Vert dP_t^{1/2} + dP^{1/2}\right\Vert_{2} =o(1)
    \end{split}
    \end{equation*}
    as $t\downarrow0$. Since $Y$ is bounded, the above display implies
    \begin{equation*}
        \lim_{t\downarrow 0} \frac{\lambda(P_t) - \lambda(P)}{t} = 
        E_P[\phi_\lambda(X) g(X)],
    \end{equation*}
    where 
    \begin{equation*}
    \begin{split}
        \phi_\lambda(X)
        &\equiv \Pi(P)^{-1}\left[
        \frac{Z}{p}
            \left\{ Y-\mu_1^Y-(\mathbf D-\mu_1^D)\lambda(P) \right\}
        -
            \frac{1-Z}{1-p}
            \left\{ Y-\mu_0^Y-(\mathbf D-\mu_0^D)\lambda(P) \right\}\right] \\
         & =\text{Cov}(Z, \mathbf D)^{-1}\left[(Z-p)( Y - \mathbf D\lambda(P) - E_P[Y -\mathbf{D}\lambda(P)])   \right].   
    \end{split}
    \end{equation*}
    and 
    $$p\equiv P[Z=1], \qquad \mu_z^Y\equiv E_P[Y\mid Z=z], \qquad \mu_z^D\equiv E_P[\mathbf D\mid Z=z].$$
    Because $\phi_\lambda\in L^2(P)$, the map $g\mapsto E_P[\phi_\lambda(X) g(X)]$ defined on $L^2(P)$ is linear and continuous. Therefore, Theorem 25.20 in \cite{van2000asymptotic} implies that $\text{Var}_P[\phi_\lambda(X)]$ is the lower bound for the variance matrix of the asymptotic variance of every regular sequence of estimators.  
    
    To conclude the proof, note that for an i.i.d.\ sample $(X_i)_{i=1}^n$ from the probability measure $P$, the 2SLS estimator has the asymptotically linear expansion
    \begin{equation*}
        \sqrt{n}[\hat\lambda - \lambda(P)] =  \frac{1}{\sqrt{n}}\sum_{i=1}^n \phi_\lambda(X_i) + o_p(1)
    \end{equation*}
    as $n\to\infty$, since $Y$ is bounded.
\end{proof}

\subsection{Finite-sample Behavior of Alternative Estimators}
\label{appendix:simulations}

This appendix reports results of a small-scale simulation study, tailored to the ISCHEMIA data, that illustrates the finite-sample performance of the 2SLS estimator in Theorem \ref{corollary:Lambda} and compares it to OLS and a g-method estimator. The target parameters are cumulative causal effects, set equal to main text estimates reported in Figure \ref{f:stacked}. For simplicity and to ease comparison across estimators, the simulation sets these as constant across observations. 

We simulate an unbalanced panel with 5 waves, mirroring the data structure of ISCHEMIA. We vary the number of initial observations across 500, 1,500, and 5,000. Inclusion in the sample at wave $w$ is determined by $$M_{w}=M_{w-1}Q_{w}, \quad Q_{w}\sim \text{Bernoulli}(q_w/q_{w-1}), \quad w=2,\ldots,5,$$ for $ M_{1}\equiv 1$, $q_1\equiv 1,$ and $(q_{w})_{w=2}^5$ chosen to mimic drop rates in ISCHEMIA. 

Each observation has a baseline outcome, $Y_{0}$, drawn with replacement from the sample of baseline SAQ summary scores in the ISCHEMIA data. We also construct a standard normal counterpart to it, denoted $y_0$.\footnote{Specifically, in each simulation, each observation is drawn with a latent baseline rank $V\sim\mathrm{Unif}[0,1]$. We generate the baseline outcome as $Y_0=\widehat F^{-1}(V)$, where $\widehat F^{-1}$ denotes the empirical quantile function of baseline SAQ summary score in ISCHEMIA. We then define $y_0=\Phi^{-1}(V)$ where $\Phi(\cdot)$ is the standard normal CDF.} The untreated potential outcome path follows: 
\begin{equation*}
    Y_{w}(0)=Y_{0}+a_w+u_{w}, \quad w=1,\ldots,5,
\end{equation*}
where $(a_w)_{w=1}^5 = (-0.5, -1, -1.5, -2, -2.5)$ is a downward deterministic trend. To induce serially correlated health shocks, the error term $u_{w}$ follows an autoregressive model:
\begin{equation*}
    u_{w}=\varphi u_{w-1}+\sigma\sqrt{1-\varphi^2}\tilde u_{w},
    \quad 
    u_{1}= \sigma \tilde u_{1},
    \quad
    \tilde u_{w}\overset{iid}{\sim} \mathcal N(0,1),
\end{equation*}
Here $\varphi=0.6$ is the persistence parameter, and $\sigma=20$ is the marginal standard deviation of $u_{iw}$. 

Next, define a latent ``prognosis index'' that determines nonrandom selection into the treatment: $$H = -\frac{\sqrt{5}}{5}\sum_{w=1}^5 \tilde u_{w}.$$ In particular, for $\beta_H=0.25,$ $\beta_y=0.5$, and $$U=\Phi\left(\beta_H H + \beta_y y_0+  [1-\beta_H^2 - \beta_y^2]^{\frac{1}{2}}\eta\right), \quad \eta\sim \mathcal N(0,1) \implies U\sim \text{Unif}[0,1],$$ the potential wave of treatment takeup when assigned to $z\in\{0,1\}$, $S(z)\in\{1,2,3,4,5,\infty\}$, is:
\begin{equation*}
     S(z)=\inf\left\{
        w\in\{1,\ldots,5\}: U\ge 1-p_{zw}
        \right\},
\end{equation*}
where $\inf\emptyset \equiv \infty$ and $p_{0w}, \, p_{1w}$ approximate empirical cumulative revascularization rates in the conservative and invasive arms in wave $w$, respectively. Exposure at wave $w$ equals $T_{w}=\max\{0,w-S+1\}$ and write $R_{wt}\equiv \mathbf1[T_w=t]$.
We then set
\begin{equation*}
    Y_{w}(t) = Y_{w}(0) + \Lambda_t,
\end{equation*}
where the $\Lambda_t$ are the estimates in Figure \ref{f:stacked}. Assignment $Z$ is drawn as a Bernoulli with probability 0.5, independently of all potential outcomes and potential treatment variables.

We consider three estimation methods for the $\Lambda_t$: the 2SLS estimator used in Figure \ref{f:stacked}, OLS for the same outcome equation, and a g-method estimator. The g-method estimates come from an additive structural nested mean model, following \cite{robins1994snm}. Specifically, we use a GMM procedure with the moment condition 
\begin{equation*}
    E\left[\mathbf1[w> s]I_{w-s}(A_{w-s} - E[A_{w-s} \mid I_{w-s}=1, X_{w-s}])\left(Y_{w} - \sum_{t=1}^5 R_{wt}\Lambda_t\right)\right]=0,
\end{equation*}
stacked across waves $w$ and $s=0,\ldots, 4$, where $A_w$ is a Bernoulli variable indicating wave $w$ is the period of revascularization, $I_w$ indicates that the period of revascularization is greater or equal to $w$, and $X_w$ is a set of covariates. For covariates, we consider assignment $Z$, the baseline variable determining selection $y_0$, one-wave lagged outcome, and difference in lagged outcomes. We specify $E[A_{w-s} \mid I_{w-s}=1, X_{w-s}]$ as a third degree polynomial in $y_{0}$, linear on lags and differenced lags, with intercepts and coefficients for $Z$ that vary with $w-s$.\footnote{For the first decision period (i.e., $w-s=1$), the lagged outcome is set to baseline and the lagged outcome difference is set to zero; for the second period, the lagged difference is the wave-1 outcome minus baseline; from the third period onward, the lagged difference is the previous outcome minus the two-period lagged outcome.} To estimate the $\Lambda_t$ using this method, we stack the above moment conditions and the OLS moment conditions for estimating $E[A_{w-s} \mid I_{w-s}=1, X_{w-s}]$; we cluster standard errors by individual. 

Table \ref{t:simulations} reports absolute bias (relative to true effects), average standard errors, 95\% coverage (i.e., the share of replications in which confidence intervals cover the true effects at the 95\% level), and power (i.e., the share of replications in which we reject a zero effect at the 5\% level) from 1,000 replications. Panel A averages over waves, while panels B and C are for waves 1 and 5, respectively. IV performs well across sample sizes and waves, with relative absolute bias no larger than 5\% and coverage around 95\% throughout. Nonetheless, it has relatively low precision for small sample sizes and, therefore, small power. OLS, in contrast, is severely biased despite high precision---yielding poor coverage. The smaller bias of the g-method estimator relative to OLS reflects that it partially accounts for selection. However, because selection into treatment depends on the ``prognosis index'', it is substantially more biased than IV. While precise in large samples, the g-method standard errors are larger than those of IV in smaller samples.

\subsection{Sensitivity to Wave Ignorability}
\label{appendix:sensitivity}

This appendix explores the sensitivity of the cumulative effects estimates in Figure \ref{f:stacked} to violations of our strong wave ignorability assumption (Assumption \ref{assump:strong_wave_ignore}). The target causal parameter is the average cumulative effect of immediate exposure for immediate compliers: for wave $w$,
\begin{equation}
\label{eq:target_sensitivity}
    \Lambda_w \equiv E[Y_w(w) - Y_w(0)\mid T_1(1) > T_1(0)],
\end{equation}
with slight abuse of notation relative to the main text. Under Assumptions \ref{assump:independence}-\ref{assump:relevance} and strong wave ignorability, Theorem \ref{corollary:Lambda} describes a 2SLS estimand that identifies the set of $\Lambda_w$ (see also the discussion in Footnote \ref{fn:bloom}). The 2SLS framework is useful for sensitivity analysis because violations of the identification assumptions map to a correlation between the instrument and the error term. The next proposition characterizes that mapping for violations of strong wave ignorability.
\begin{proposition}
\label{prop:sensitivity}
    Suppose Assumptions \ref{assump:independence}-\ref{assump:relevance} hold and let $R_{wt}\equiv \mathbf1[T_w=t]$ for each $t\in\{1,\ldots,\bar w\}$. For each $w\in\{1,\ldots,\bar w\}$, consider the linear model
    \begin{equation}
    \label{eq:lm_sensitivity}
        Y_w = \tilde\phi + \sum_{j=2}^{\bar w}\tilde\delta_j\mathbf{1}[j=w] + \sum_{t=1}^{\bar w} \Lambda_t R_{wt} +  \gamma_w Z + \tilde\varepsilon_w,
    \end{equation}
    where the set of $\Lambda_t$ is as in \eqref{eq:target_sensitivity}, $\Lambda_0\equiv0\equiv \gamma_1$,
    \begin{equation*}
    \begin{split}
        & \gamma_w \equiv 
        \sum_{0\le t'<t\le w}
        E[(Y_w(t)-Y_w(t')-(\Lambda_{t}-\Lambda_{t'}))
        \mathbf1[T_w(1)=t,T_w(0)=t']],
        \quad w\in\{2,\ldots, \bar w\}, \\
        &\tilde\phi \equiv E[Y_1 - \Lambda_1 R_{11}], \\
        & \tilde\delta_j \equiv E\Bigg[Y_j - \sum_{t=1}^{\bar w} \Lambda_t R_{jt} -\gamma_j Z  \Bigg] - E[Y_1 - \Lambda_1 R_{11}], \quad j\in\{2,\ldots,\bar w\}.
    \end{split}
    \end{equation*}
    Then, $E[\tilde\varepsilon_w\mid Z]=0$ for all $w\in\{1,\ldots, \bar w\}$. 
\end{proposition}
\begin{proof}
    For $w=1$, the model simplifies to $Y_1 = \tilde\phi + \Lambda_1 R_{11} + \gamma_1 Z + \tilde\varepsilon_1$. Because $\Lambda_1$ is the average effect for immediate compliers in the first wave and $R_{11}$ is the Bernoulli treatment indicator, the \cite{late94} LATE Theorem implies $E[\tilde\varepsilon_w\mid Z]=0$ for $\tilde\phi= E[Y_1 - \Lambda_1 R_{11}]$ and $\gamma_1=0$. 

    Now fix $w\in\{2,\ldots,\bar w\}$. By definition of $\tilde\phi$ and $\tilde\delta_w$, $E[\tilde\varepsilon_w] =0$. Because $Z$ is binary, it suffices to show that $E[Z\tilde\varepsilon_w]=0$, which is equivalent to
    \begin{equation*}
        \gamma_wE[Z] = E\Bigg[Z\Bigg(Y_w -\tilde\phi - \tilde\delta_w -  \sum_{t=1}^{\bar w} \Lambda_t R_{wt}  \Bigg)\Bigg].
    \end{equation*}
    It then follows from Assumptions \ref{assump:independence} and \ref{assump:relevance}, $Y_w=Y_w(0) + \sum_{t=1}^w (Y_w(t) - Y_w(0))R_{wt}$, and the definition of $\tilde\psi,\tilde\delta_w$ that $E[Z\tilde\varepsilon_w]=0$ if, and only if, 
    \begin{equation*}
    \begin{split}
        \gamma_w\text{Var}[Z] 
        =& E\Bigg[(Z-E[Z]) \Bigg(\sum_{t=1}^w (Y_w(t) - Y_w(0) - \Lambda_t)R_{wt}\Bigg)\Bigg] \\
        =& E[(Z-E[Z])Z]E\Bigg[\sum_{t=1}^w (Y_w(t) - Y_w(0) - \Lambda_t)\mathbf1[T_w(1)=t]\Bigg] \\
        &+E[(Z-E[Z])(1-Z)]E\Bigg[\sum_{t=1}^w (Y_w(t) - Y_w(0) - \Lambda_t)\mathbf1[T_w(0)=t]\Bigg].
    \end{split}
    \end{equation*}
    Therefore, Assumption \ref{assump:exposure_monotonicity} implies that $E[Z\tilde\varepsilon_w] =0$ if, and only if,
    \begin{equation*}
    \begin{split}
        \gamma_w 
        &= E\Bigg[\sum_{t=1}^w (Y_w(t) - Y_w(0) - \Lambda_t)\mathbf1[T_w(1)=t]\Bigg] -  E\Bigg[\sum_{t=1}^w (Y_w(t) - Y_w(0) - \Lambda_t)\mathbf1[T_w(0)=t]\Bigg] \\
        & = \sum_{0\le t'<t\le w}
        E[(Y_w(t)-Y_w(t')-(\Lambda_{t}-\Lambda_{t'}))
        \mathbf1[T_w(1)=t,T_w(0)=t']].
    \end{split}
    \end{equation*}
\end{proof}

Proposition \ref{prop:sensitivity} parameterizes violations of strong wave ignorability in terms of the $\gamma_w$ coefficients, in line with earlier sensitivity methods for linear instrumental variable models (e.g., \cite{conley12plausibly, mastenpoirier_salvaging, cinellihazlett}). Specifically, $\gamma_w$ is generally nonzero when wave-$w$ compliers moved from some treatment margin $t^\prime$ to another margin $t$ (with $T_w(1)=t$ and $T_w(0)=t^\prime$) have systematically different incremental effects than that of immediate compliers, $\Lambda_t-\Lambda_{t^\prime}$. The proposition shows that these deviation parameters arise as direct effects of the instrument $Z$ in the stacked model from Theorem \ref{corollary:Lambda}.

We apply this result to the general sensitivity framework of \cite{conley12plausibly} (henceforth, CHR) to explore robustness to nonzero $\gamma_w$. The CHR approach builds on the observation that if the vector of $\gamma_w$ was known, a 2SLS of the modified outcome $Y_w-\gamma_w Z$ on the set of $R_{wt}$ instrumented by $Z$ and $Z\times\mathbf1[j=w],j\in\{2,...,\bar{w}\}$ would estimate the set of $\Lambda_t$. In practice, however, the set of $\gamma_w$ is unknown. We follow CHR in postulating a bounded set of possible values for the vector of $\gamma_w$, $\Gamma\subseteq\mathbb{R}^{\bar w}$, with $\gamma_1\equiv0$. Then, for each point in $\Gamma$, we can construct confidence intervals for each $\Lambda_t$. Taking the union of the confidence intervals over $\Gamma$ guarantees the prescribed coverage as long as the true vector of $\gamma_w$ lies on $\Gamma$. CHR approximate such a union of confidence intervals by discretizing $\Gamma$ and taking unions over the grid points.

We consider two approaches to specify $\Gamma$, both parametrized by bounds on the possible incremental effect deviations underlying each $\gamma_w$:
\begin{equation*}
   \Delta_w^{t,t'}\equiv E[Y_w(t) - Y_w(t')\mid T_w(1)=t, T_w(0)=t'] - (\Lambda_t - \Lambda_{t'}).
\end{equation*}
Specifically, we assume that $|\Delta_w^{t,t'}|$ is at most a $\rho\in[0,1]$ fraction of the absolute difference in cumulative effects:
\begin{equation}
    \label{eq:boundonDelta}
    |\Delta_w^{t,t'}|\le \rho|\Lambda_t - \Lambda_{t'}|.
\end{equation}
Our first approach is conservative. Because the probabilities $P[T_w(1)=t, T_w(0)=t']$ are not identified, it takes the worst-case by setting $|\gamma_w| \le \max\{\rho|\Lambda_t - \Lambda_{t'}|: 0\le t'<t\le w\}.$ The second approach instead sets $$|\gamma_w|\le\sum_{0\le t<w} (E[\mathbf1[T_w=t] \mid Z=0] - E[\mathbf1[T_w=t] \mid Z=1])\rho|\Lambda_w - \Lambda_t|,$$
which are valid bounds on $\gamma_w$ under IMCO (Assumption \ref{assump:imco}) and \eqref{eq:boundonDelta} since then
\begin{equation*}
    \begin{split}
       & \gamma_w=\sum_{0\le t<w}P[T_w(1)=w, T_w(0)=t]\Delta_w^{w,t}, \\
       & P[T_w(1)=w, T_w(0)=t'] = E[\mathbf1[T_w=t] \mid Z=0] - E[\mathbf1[T_w=t] \mid Z=1].
    \end{split}
\end{equation*}
In practice, because the target parameters $\Lambda_t$ are not identified when strong wave ignorability doesn't hold, we use the estimates in Figure \ref{f:stacked} to bound $|\Delta_w^{t,t'}|$. This yields bounds that are compatible with the magnitude of effects in the ISCHEMIA setting.

Figure \ref{f:sensitivity} reports union-of-confidence-interval estimates under the two sensitivity calibrations, for $\rho\in\{0,0.25,0.5\}$ (the $\rho=0$ results replicate the baseline Figure \ref{f:stacked} estimates). The results suggest that the estimated average cumulative effects of revascularization are not overly sensitive to violations of wave ignorability, especially for the SAQ summary score. Under the IMCO-based calibration, the data supports the main results for violations up to $\rho=0.25$. At this value, the null of no effect is rejected at the 5 percent level for both outcomes in all follow-up waves except the last. For the SAQ summary score, the first three follow-up effects remain statistically different from zero even when $\rho=0.5$. The more conservative max-bound calibration leads to wider intervals, as expected, while still suggesting positive effects on the SAQ summary score at intermediate follow-up horizons. Taken together, these results build confidence in the robustness of our baseline findings.

\clearpage

\begin{figure}[H]
    \caption{Differences in Exposure CDF by Treatment Assignment Status}
    \label{f:stoch_dom}
    \centering
    \begin{minipage}{16.5cm}
        {\centering{\includegraphics[width=0.65\linewidth]{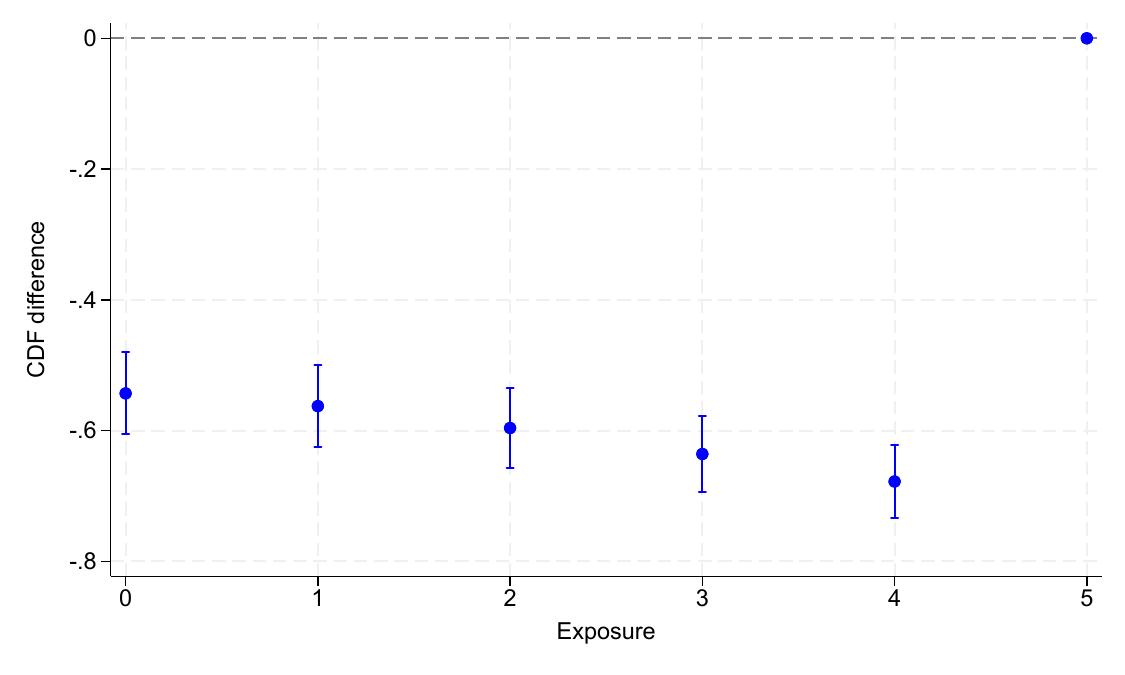}\par}}
        \footnotesize \textit{Notes:} This figure plots estimated differences in CDFs between the treatment and control groups for exposure levels $T_w$. For each exposure year $t$, the estimate comes from the regression of the indicator $1[T_w \le t]$ on treatment assignment status $Z$. Vertical bars denote the 95 \% robust confidence intervals. This figure is based on observations from wave 5.  \par
    \end{minipage}
\end{figure}

\pagebreak
\begin{figure}[H]
    \caption{Permitted and Ruled-Out Paths for Incremental Effects in \cite{Bowden2025}}
    \label{f:bowden}
    \centering
    \begin{minipage}{16.5cm}
        {\centering{\includegraphics[width=0.65\linewidth]{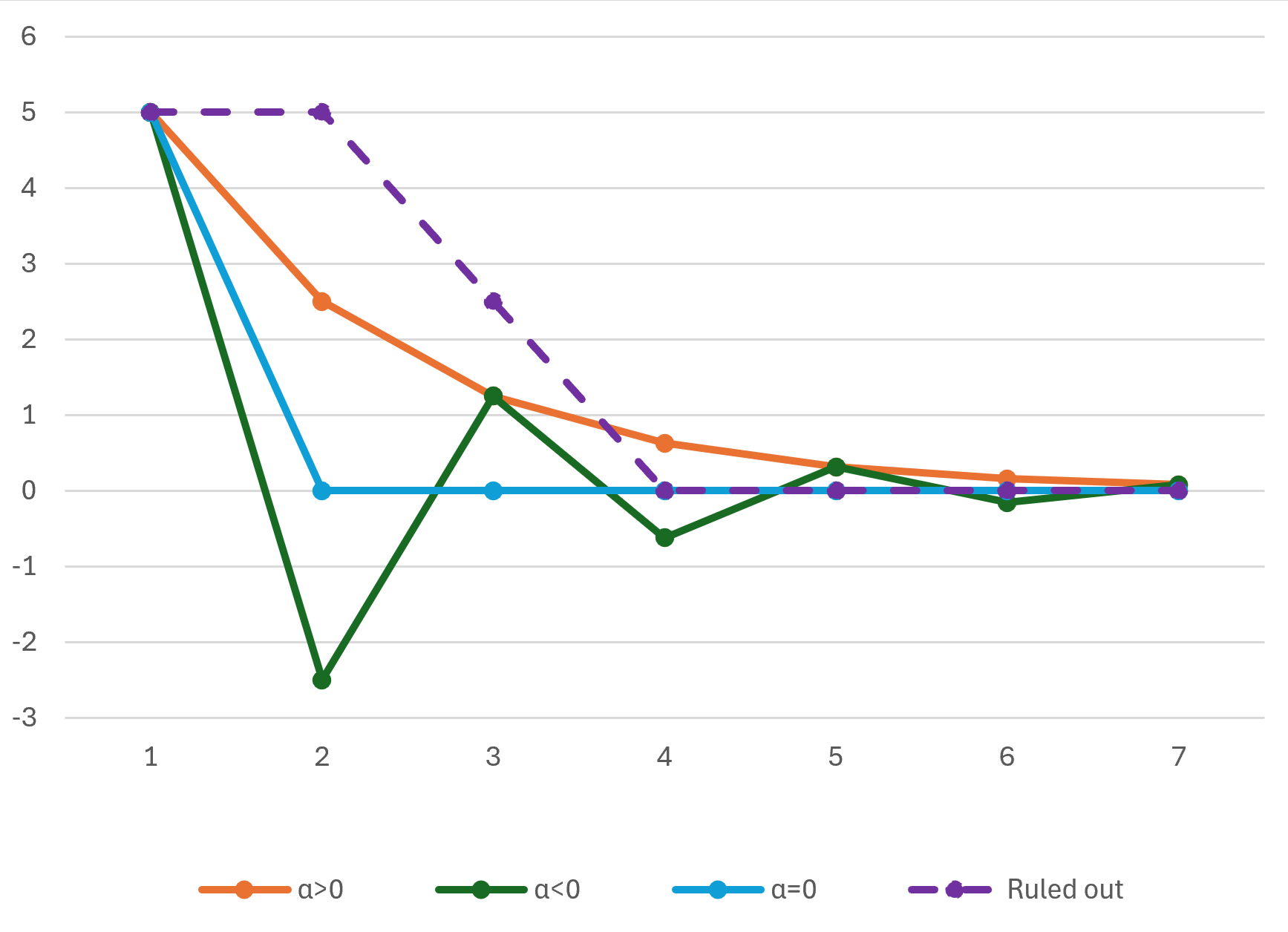}\par}}
        \footnotesize \textit{Notes:} The first three series in this figure illustrate possible paths for incremental effects $\lambda_t$ under the \cite{Bowden2025} parameterization, \eqref{eq:bowden_geometric}, for $\beta=5$. The first series sets $\alpha=0.5$, the second series sets $\alpha=-0.5$ and the third series sets $\alpha=0$. The fourth series in this figure is ruled out by Equation \eqref{eq:bowden_geometric}.\par
    \end{minipage}
\end{figure}

\pagebreak
\begin{figure}[H]
    \caption{Sensitivity Analysis for Average Cumulative Causal Effects}
    \label{f:sensitivity}
    \centering
    \begin{minipage}{16.5cm}
        {\centering{\includegraphics[width=\linewidth]{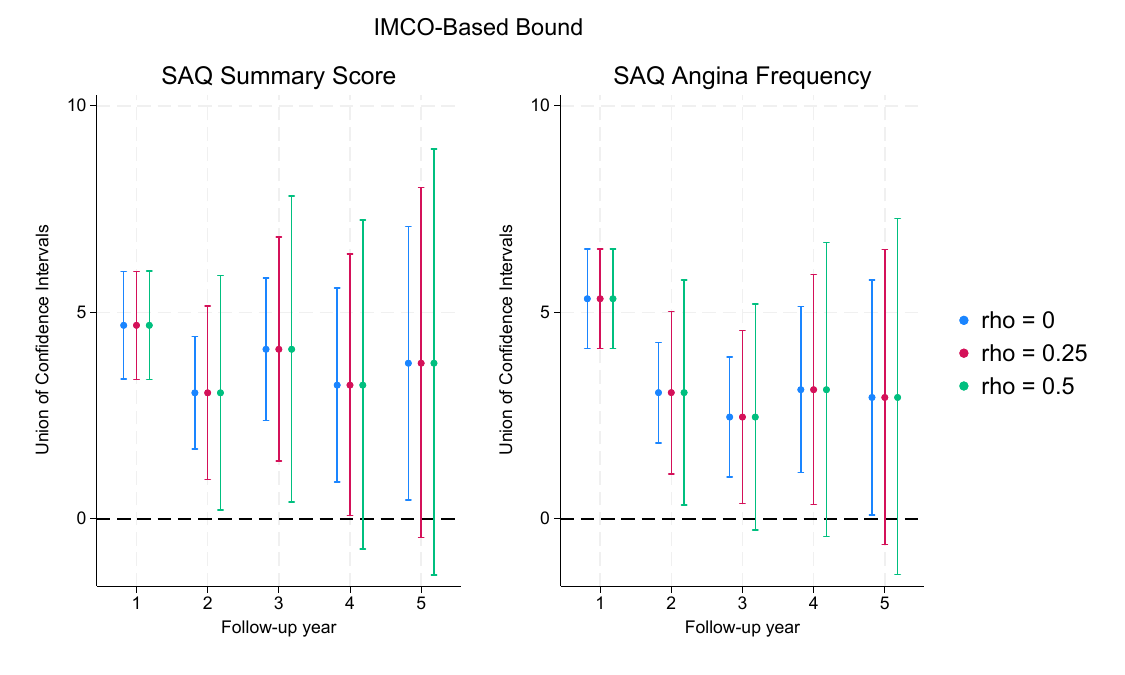}}\par}
        {\centering{\includegraphics[width=\linewidth]{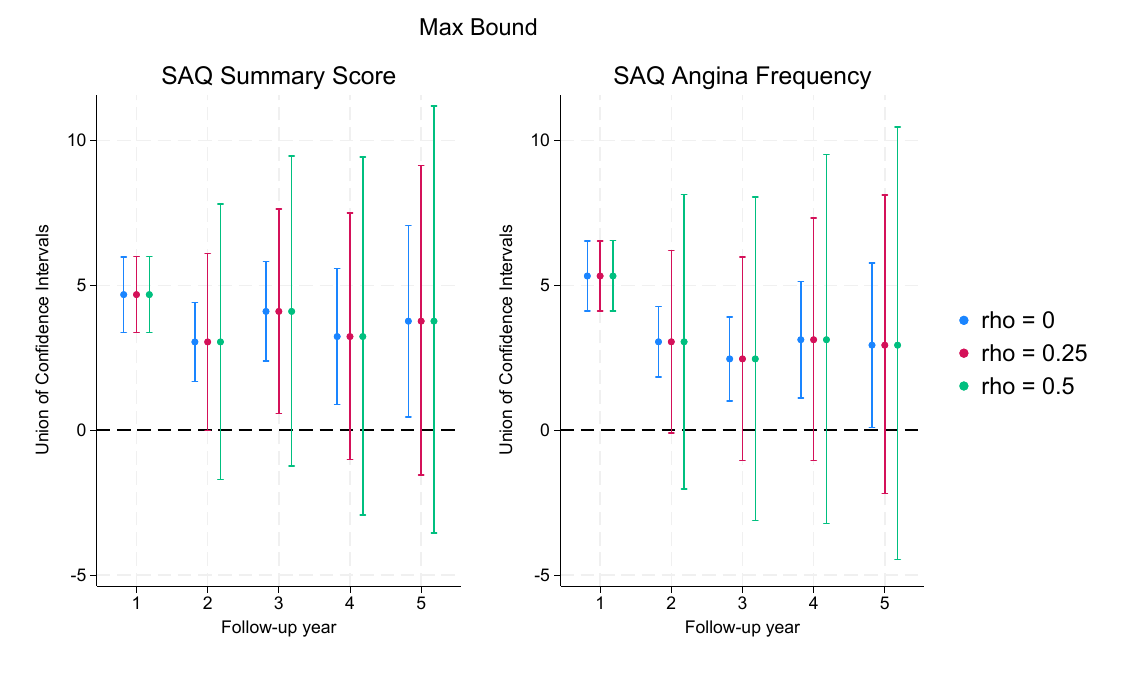}}\par}
        \footnotesize \textit{Notes:} This figure reports sensitivity results for the average cumulative effect of immediate exposure on SAQ Summary Score and SAQ Angina Frequency for immediate compliers. The top panel shows the IMCO-based results for $\rho\in\{0,0.25,0.5\}$, while the bottom panel shows results based on max-bounds for the same values of $\rho$. For each value of $\rho>0$, the set of possible values for the vector of $\gamma_w$ is discretized using 5 grid points for each $\gamma_w$, except $\gamma_1$ which is set to zero. The point estimates reported here equal the ones in the main text (Figure \ref{f:stacked}). Appendix \ref{appendix:sensitivity} describes the sensitivity analysis procedure in detail.  \par
    \end{minipage}
\end{figure}

\clearpage
\begin{table}[H]
    \caption{Comparison with \cite{Bowden2025}}
    \label{t:bowden}
        \centering
    \begin{minipage}{6.5in}
        {\centering\includegraphics[width=16cm]{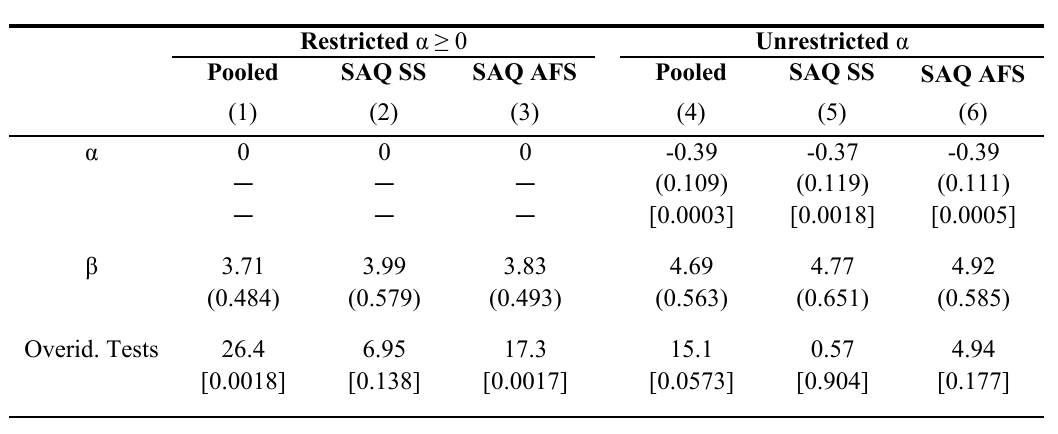}\par}
        \footnotesize \textit{Notes:} This table reports efficiently-weighted two-step GMM estimates based on the moment condition in Equation \eqref{eq:bowden_ischemia}. Column 1 reports the results from a model imposing $\alpha=0$ when the outcomes, SAQ Summary Score and SAQ Angina Frequency Score, are pooled. Columns 2 and 3 report results using the same model for each outcome separately. Columns 4, 5, and 6 report results for the same outcomes, respectively, for a model where $\alpha$ is unrestricted. Estimates were computed with controls for baseline angina frequency scores and enrollment region. Individual-clustered standard errors are reported in parentheses, and p-values for overidentification tests are reported in square brackets. \par 
    \end{minipage}
\end{table}

\pagebreak
\begin{table}[H]
    \caption{Simulation Results Comparing IV, OLS, and G-method}
    \label{t:simulations}
        \centering
    \begin{minipage}{6.5in}
        {\centering\includegraphics[width=11cm]{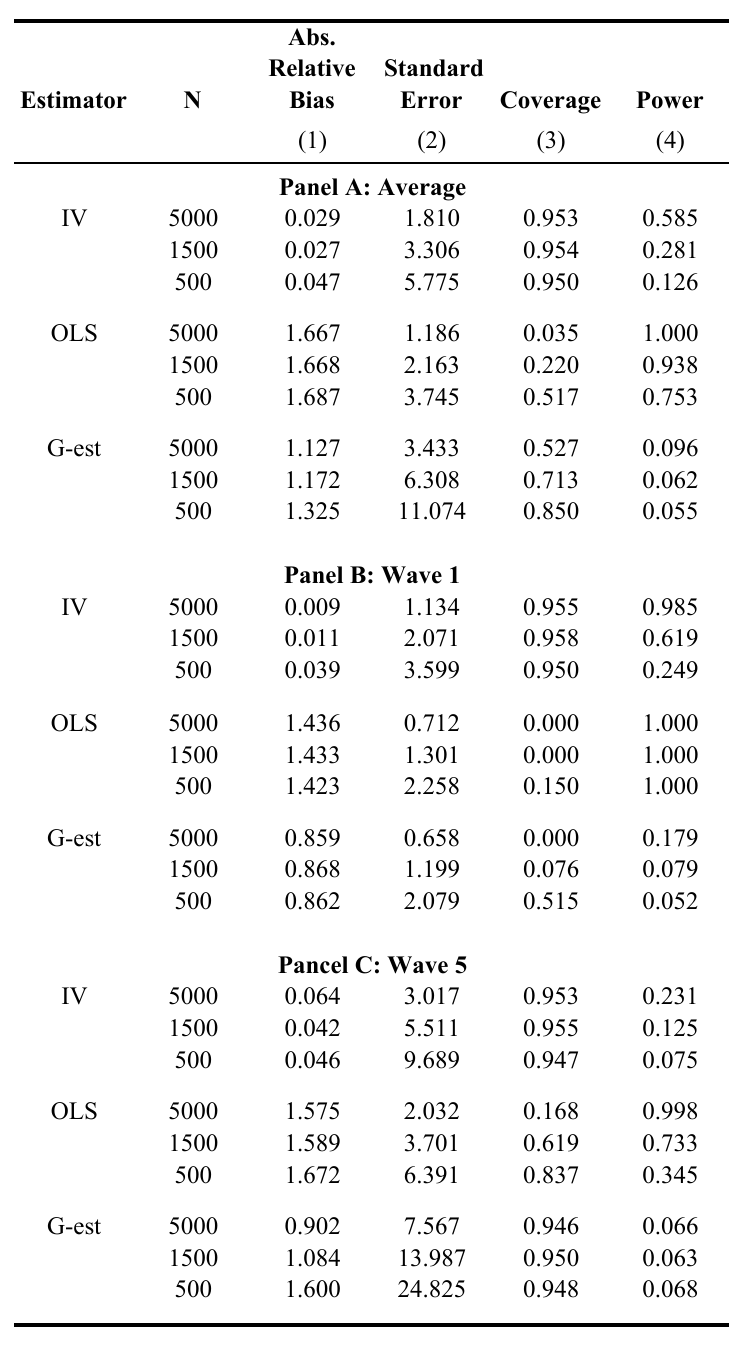}\par}
        \footnotesize \textit{Notes:} This table reports simulation results comparing the finite sample performance of IV, OLS, and g-method for 1,000 replications. Sample sizes considered are 500, 1,500, and 5,000. Column 1 reports absolute bias relative to true causal effects. Column 2 reports standard errors averaged across replications. Column 3 reports 95\% nominal coverage averaged across replications (i.e., the share of replications in which confidence intervals cover the true effects at the 95\% level). Column 4 reports the share of replications in which the null of no effect is rejected at the 5\% level. Panel A reports results averaged across waves, while panels B and C report results for waves 1 and 5, respectively. Appendix \ref{appendix:simulations} details the data-generating process and estimators. \par 
    \end{minipage}
\end{table}

\end{document}